\documentclass[11pt]{article}

\usepackage[active]{srcltx}
\usepackage{hyperref}
\usepackage{amsmath,amssymb}
\usepackage[all]{xy}
\usepackage{graphicx}
\usepackage{color}
\usepackage{tikz-cd}
 1 
\newtheorem{theorem}{Theorem}

\newtheorem{definition}{Definition}

\newtheorem{example}[definition]{Example}

\def\qed{\ifvmode\Realemovelastskip\fi
{\unskip\nobreak\hfil\penalty50\hbox{}\nobreak\hfil \hbox{\vrule
height1.2ex width1.2ex}\parfillskip=0pt \finalhyphendemerits=0
\par\smallskip}}
\def\qedr{\ifvmode\Realemovelastskip\fi
{\unskip\nobreak\hfil\penalty50\hbox{}\nobreak\hfil \hbox{
$\diamond$}\parfillskip=0pt \finalhyphendemerits=0
\par\smallskip}}
\def\ds{\displaystyle}
\newenvironment{proof}{\noindent{\sl Proof:~~~}}{\quad \qed}

\def\beq{\begin{equation}}
\def\eeq{\end{equation}}
\def\bea{\begin{eqnarray}}
\def\eea{\end{eqnarray}}
\def\beann{\begin{eqnarray*}}
\def\eeann{\end{eqnarray*}}
\def\beasn{\begin{sneqnarray}}
\def\eeasn{\end{sneqnarray}}
\def\ben{\begin{enumerate}}
\def\een{\end{enumerate}}
\def\bit{\begin{itemize}}
\def\eit{\end{itemize}}

\def\derpar#1#2{\displaystyle\frac{\partial{#1}}{\partial{#2}}}

\def\restric#1#2{\left.#1\right|_{#2}}

\def\vf{{\mathfrak{X}}}

\def\Real{\mathbb{R}}
\def\R{\mathbb{R}}

\def\Tan{{\rm T}}

\def\qed{\ifvmode\removelastskip\fi
{\unskip\nobreak\hfil\penalty50\hbox{}\nobreak\hfil \hbox{\vrule
height1.2ex width1.2ex}\parfillskip=0pt \finalhyphendemerits=0
\par\smallskip}}

\begin{document}

\centerline{\bf A Hamilton-Jacobi formalism for higher order implicit systems}

\vskip 0.2cm
\vskip 0.5cm
\vskip 0.5cm

\centerline{O\u{g}ul Esen$^{\dagger}$, Manuel de Le\'on$^{\ddagger}$, Cristina
Sard\'on$^{*}$}
\vskip 0.5cm

\centerline{Department of Mathematics$^{\dagger}$}
\centerline{Gebze Technical University}
\centerline{41400 Gebze, Kocaeli, Turkey.}
\vskip 0.5cm

\centerline{Instituto de Ciencias Matem\'aticas, Campus Cantoblanco$^{\ddagger}$}
\centerline{Consejo Superior de Investigaciones Cient\'ificas}
\centerline{Real Academia Espa{\~n}ola de las Ciencias}
\centerline{C/ Nicol\'as Cabrera, 13--15, 28049, Madrid. SPAIN}
\vskip 0.5cm

\centerline{Instituto de Ciencias Matem\'aticas, Campus Cantoblanco$^{*}$}
\centerline{Consejo Superior de Investigaciones Cient\'ificas}
\centerline{C/ Nicol\'as Cabrera, 13--15, 28049, Madrid. SPAIN}

\begin{abstract}
In this paper, we present a generalization of a Hamilton--Jacobi theory to higher order implicit differential equations. We propose two different backgrounds to deal with higher order implicit Lagrangian theories: the Ostrogradsky approach and the Schmidt transform, which convert a higher order Lagrangian into a first order one. The Ostrogradsky approach involves the addition of new independent variables to account for higher order derivatives, whilst the Schmidt transform adds gauge invariant terms to the Lagrangian function. 
In these two settings, the implicit character of the resulting equations will be treated in two different ways in order to provide a Hamilton--Jacobi equation. On one hand, the implicit differential equation will be a Lagrangian submanifold of a higher order tangent bundle and it is generated by a Morse family. On the other hand, we will rely on the existence of an auxiliary section of a certain bundle that allows the construction of local vector fields, even if the differential equations are implicit.
We will illustrate some examples of our proposed schemes, and discuss the applicability of the proposal.

\textbf{Keywords}: Hamilton--Jacobi, second order Lagrangians, Schmidt-Legendre transformation; Ostragradsky-Legendre transformation; implicit differential equations;  Morse families; higher order Lagrange equations; constrained Lagrangians.
\end{abstract}

%
%
%
%
%


\section{Introduction}
\label{sec:Introduction}

The Hamilton--Jacobi theory (HJ theory) is a very useful theory for the study of dynamical systems. It is rooted in the idea of finding an appropriate canonical
transformation \cite{Arn,JS} that leads the system to equilibrium with a trivial Hamiltonian, and pairs of action-angle variables that render the dynamics trivial.
This philosophy has brought many interesting results, deriving into integrability theories, reduction, KAM theory, among others \cite{Gomis1,EJLS,sorren,Su}.
Our interest resides in the geometric interpretation of this theory \cite{Ab,LM87,MMM}, its formulation, and applications. See for example
\cite{art:Carinena_Gracia_Marmo_Martinez_Munoz_Roman06}, where a
geometric framework for the HJ theory was presented
and the Hamilton--Jacobi equation was formulated both in the
Lagrangian and in the Hamiltonian formalisms of autonomous and
nonautonomous mechanics.

\bigskip

\noindent\textbf{{Geometric HJ theory.}} The time-independent HJ theory is a partial differential equation for a generating function $W$,
 \begin{equation}\label{hje}
H\left(q^{A},\frac{\partial W}{\partial q^{A}}\right)=E
 \end{equation}
on a $n$-dimensional configuration space $Q$ with local coordinates $(q^{A})$, and $E$ is the energy. This generating function comes from a separable generating function in $t$, i.e., $G_2=W(q^{A})-Et$ where the momentum in the canonical transformation corresponds with $p_A=\partial W/\partial q^A$. 
The equation \eqref{hje} can be interpreted geometrically in the following way. First, consider a Hamiltonian vector field $X_H$ on $T^{*}Q$, and a one-form $\gamma$ on $Q$.
We define a vector field $X_H^{\gamma}$ on $Q$ by
\begin{equation}\label{gammarelated}
 X_H^{\gamma}=T\pi\circ X_H\circ \gamma.
\end{equation}
This definition implies the commutativity of the following diagram.
\begin{equation}\label{Xgg}
  \xymatrix{ T^{*}Q
\ar[dd]^{\pi_{Q}} \ar[rrr]^{X_H}&   & &TT^{*}Q\ar[dd]^{T\pi_{Q}}\\
  &  & &\\
 Q\ar@/^2pc/[uu]^{\gamma}\ar[rrr]^{X_H^{\gamma}}&  & & TQ }
\end{equation}
We state the geometric Hamilton-Jacobi theorem as follows.

\begin{theorem} \label{HJT}
Assume that we have a closed one-form $\gamma=dW$ on $Q$, then the following two conditions are equivalent:

\begin{enumerate}
\item The vector fields $X_{H}$ and $X_{H}^{\gamma }$ are $\gamma$-related, that is

\begin{equation}
T\gamma(X^{\gamma})=X\circ\gamma.
\end{equation}
\item and the following equation is fulfilled
$$d\left( H\circ \gamma \right)=0.$$
\end{enumerate}
Under these conditions, we will say that $\gamma$ is a solution of the Hamilton--Jacobi problem.
\end{theorem}

Therefore, a HJ theory finds solutions on the lower dimensional manifold $Q$ and retrieves them on the higher dimensional manifold $T^{*}Q$ by the existence of a section $\gamma$ of the cotangent bundle which is the solution $\gamma$ of the Hamilton--Jacobi equation \eqref{hje}.

This picture \eqref{Xgg} can be devised in different situations, as it is the case of nonholonomic systems \cite{HJTeam2,EJLS,leones1,LOS-12,leones2,blo,OFB-11}, geometric
mechanics on Lie algebroids \cite{BMMP-10} and almost-Poisson
manifolds, singular systems \cite{LMV-12}, Nambu-Poisson framework \cite{LSNambu}, control
theory \cite{BLMMM-12}, classical field theories
\cite{LMM-09,LMMSV-12,DeLeon_Vilarino}, partial differential equations in general
\cite{Vi-11}, the geometric discretization of the Hamilton--Jacobi
equation \cite{LSdis,OBL-11}, and others \cite{BLM-12,HJteam-K}.

 One of the conditions in this theory for $\gamma$  to be a solution to be a solution of the HJ equation is that its image Im$(\gamma)$ is a Lagrangian submanifold.
Recall that given a symplectic manifold as the pair $(M,\Omega)$, a sufficient condition for a submanifold $S\subset M$ to be a Lagrangian submanifold is that $TS=TS^{\bot}$ is satisfied.
 If $S$ is a isotropic subspace of a symplectic manifold $(M,\Omega)$, then $S$ is Lagrangian if an only if $\text{dim}S=(\text{dim}M)/2$, this is equivalent to saying that $\gamma$ is closed on a manifold that is a cotangent bundle (this would not be true for an arbitrary manifold). For different types of manifolds (Poisson, Nambu--Poisson, etc), the definition of a Lagrangian submanifold has been accommodated to its background. See for example \cite{LM87}. In the case of mechanical systems, the Lagrangian submanifolds have a physical interpretation as a generalization of the set of possible initial momenta of a given point in the configuration space.  As it is well known, the image space of a closed one-form $\gamma$ is a Lagrangian submanifold, and the Weinstein tubular neighborhood theorem \cite{weinsteintubular}
 assures there exists a tubular neighborhood of a Lagrangian submanifold $S$ in $(M,\Omega)$ that is symplectomorphic to an open neighborhood of the zero section in $(T^{*}S,\Omega_S)$. 
 
Given a Hamiltonian system $(T^*Q,H,\Omega_Q)$, the image space of the Hamiltonian vector field $X_H$ is a Lagrangian submanifold of the symplectic manifold $TT^*Q$ equipped with a symplectic structure $\Omega^T_Q$ computed to be the complete lift of $\Omega_Q$. 
Accordingly, a differential system is said to be a Hamiltonian system if it can be recast as a Lagrangian submanifold of a certain symplectic manifold. It is evident that if the Lagrangian submanifold is horizontal, according to the Poincar\'{e} lemma \cite{LM87}, there locally exists a Hamiltonian function generating the dynamics. This results with an explicit differential system, called as an explicit  Hamiltonian system. Otherwise, if a Lagrangian submanifold is non-horizontal then it is not guarantied the existence of a  Hamiltonian vector field generating the dynamics. We call such kind of systems as implicit Hamiltonian system.

\bigskip 
 
\noindent\textbf{Hamilton-Jacobi theory for implicit hamiltonian systems.} 
It is important to remark that the classical HJ theory only deals with explicit Hamiltonian systems. In \cite{ELS}, we formulated an idea to deal with the implicit character. We considered implicit first order differential equations as a submanifold $S$ of $TT^{*}Q$. This submanifold $S$ projects to $TQ$ by the tangent mapping $T\pi_{Q}$ to a submanifold $S^{\gamma}=T\pi_{Q}(S)$ of $TQ$, which is another implicit differential equation on $Q$ and $\gamma:Q\rightarrow T^{*}Q$ would be the solution for the implicit Hamilton-Jacobi problem.
The philosophy of the geometric HJ theory is to
retrieve solutions of $S$, provided the solutions of $T\pi_{Q}(S)$. Let us picture this in a diagram:

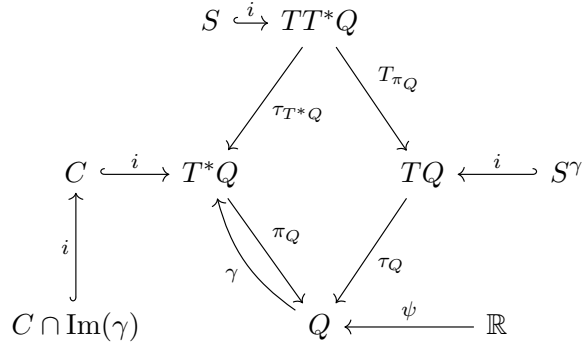
\begin{center}
 \begin{tikzcd}[column sep=tiny,row sep=huge]
  & &  &S\arrow[r, hook, "i"] &TT^*Q \arrow[ld, "\tau_{T^*Q}"] \arrow[rd, "T_{\pi_Q}"] & & \\
  & &C\arrow[r, hook, "i"] &T^*Q\arrow[dr,"\pi_{Q}"] & &TQ\arrow[dl,"\tau_{Q}"]  &&S^{\gamma}\arrow[ll, hook', "i"']\\
& &C\cap \text{Im}(\gamma)\arrow[u, hook, "i"] & &Q\arrow[lu,"\gamma",bend left=20] & &\mathbb{R}\arrow[ll,"\psi"']
 \end{tikzcd}
\end{center}

In similar fashion
as in the classical Hamilton-Jacobi theorem \eqref{HJT}, $(Q,T^{*}Q)$
is a symplectic manifold with a canonical two form $\Omega_Q$ and canonical projections $\pi_Q:T^{*}Q\rightarrow Q$ and $\tau_Q:TQ\rightarrow Q$. In order to lift the solutions in $Q$ to $T^{*}Q$, we are still in need of a closed one-form $\gamma$ on $Q$, but two ingredients
of the theory are missing. One is that the base manifold $C=\tau_{T^{*}Q}(S)$ is not necessarily the whole $T^{*}Q$, but possibly a proper submanifold of it.
The second is the nonexistence of a Hamiltonian vector field due to the
implicit character of the equations. In the classical theory, the major role
of the Hamiltonian vector field is to connect the image space of $\gamma$ and the submanifold $S$. In this case, the Lagrangian submanifold has failed to be horizontal, and a Hamiltonian function does not exist even in a local chart.

Our first idea to work with a nonhorizontal submanifold is to make use of the Maslov-H\"ormander theorem (known also as generalized Poincar\'e lemma) \cite{BeTulcz, jane,LM87,weinsteintubular}, which affirms that there exists a Morse function (a family of generating functions) that generates the dynamics of the implicit system. This Morse function plays the role of the Hamiltonian in the explicit picture. The Morse family is defined on the total space of a smooth bundle linked to $TT^{*}Q$ by means of a special symplectic structure. 

The second idea to deal with the implicit character of the system is based on the local construction of a vector field. For this construction we need to consider a second auxiliary section $\sigma:C\cap Im(\gamma)\rightarrow E$, because since $E$ is implicit, there may exist several vectors in $E$ projecting to the
same point, say $c$, in $C$. The role of the section $\sigma$ is to reduce this unknown number to one. As a result, we arrive at a vector field $X_{\sigma}$ that satisfies Hamilton-like equations. We will show details in forthcoming sections applying this for implicit higher order systems \cite{ELS}.

\bigskip

\noindent \textbf{Higher order systems.} 
Higher order systems are not so common as first and second order systems, at least in the physical literature, but they still make an appearance in the mathematical description of relativistic particles with spin, string theories,
gravitation, Podolsky's electromagnetism,
in some problems of fluid mechanics and classical physics, and
in numerical models arising from the geometric discretization of first order dynamical systems
(see \cite{PRR,PRR2}
for a long but non-exhaustive list of references). In \cite{PRR} the authors propose an unified formalism for autonomous higher order dynamical systems in the Lagrangian and Hamiltonian counterparts; and in \cite{CLPRR} they propose a Hamilton--Jacobi theory for higher order systems that are explicit.

Nonetheless, these previous works had not considered the possible implicit character of the equation coming from the Lagrangian or Hamiltonian part when they give rise to higher order equations. It is important to depict a Hamilton--Jacobi theory for implicit systems, given the number of singular Lagrangians in the sense of Dirac-Bergmann \cite{AnderBerg,Dirac2}, including
systems appearing in gauge theories \cite{book:DeLeon_Rodrigues85}. The  Euler--Lagrange equations for these systems give rise to differential equations that are implicit, and they cannot be put in a normal form.
The geometric formalism for dealing with dynamical systems in their implicit form and of lower order was introduced in \cite{tulcz1,tulcz2}, 
where a unified approach for the Lagrangian description of (time-independent) constrained mechanical systems is provided through a technique that generates implicit differential equations on $T^{*}Q$ from one-forms defined on the total space of any fiber bundle over $TQ$ \cite{BaGrasMende}. Other authors designed other methods following Dirac-Bergmann prescription to be able to deal with singular Hamiltonian and Lagrangian theories, e.g., the  Gotay--Nester algorithm \cite{gotay2,gotay,gotay3,gotay0}. 

\bigskip

\noindent \textbf{Goal of the paper and the contents.}
In this paper, we construct a Hamilton--Jacobi theory for higher order mechanical systems described through implicit differential equations. More concretely, we will generalize the geometry that we have proposed in \cite{ELS}. As we have mentioned previously, in \cite{ELS}, two ideas are proposed. In the present work, we generalize the first idea by considering a higher order implicit differential equation as a submanifold $S$ of $TT^{*}T^{k-1}Q$ generated by a Morse function $F$ defined on the Whitney sum $T^{*}T^{k-1}Q\times_{T^{k-1}Q}T^kQ$ and the projected submanifold is $S^{\gamma}$ in $TT^{k-1}Q$. For a higher order Lagrangian function on $T^kQ$, whether being degenerate or nondegenerate, we will take the associated energy function as the Morse family. The idea is to find a section $\gamma:T^{k-1}Q\rightarrow T^{*}T^{k-1}Q$, in other words, it is a one-form on $T^{k-1}Q$ such that for a solution $\psi:\mathbb{R}\rightarrow T^{k-1}Q$ of $S^{\gamma}$, we have that $\gamma\circ \psi:\mathbb{R}\rightarrow T^{*}T^{k-1}Q$ is a solution of $S$. If the Lagrangian is nondegenerate, the energy function can be written in terms of the momentum coordinates available in the iterated cotangent bundle. This results with a well-defined Hamiltonian function on $T^*T^{k-1}Q$ \cite{Ostro}. In this case, the Lagrangian submanifold is horizontal with respect to $\pi_{T^{k-1}Q}$. As discussed previously, the classical Hamilton-Jacobi theory is proper for such systems. If the Lagrangian is degenerate, then the energy function still remains to be a Morse family but, in this case, a well-defined Hamiltonian function defined on the total space of the iterated cotangent bundle is not possible. In geometrical terms, the Lagrangian submanifold is nonhorizontal. This is an implicit Hamiltonian system and we are in the realm of implicit HJ theory. We also present the local construction of a vector field in order to provide an equivalent theory to \eqref{HJT} for implicit higher order systems.

On the other hand, there are two methods to write higher order Lagrangians into the form of first order Lagrangians, namely Ostrogradsky, and Schmidt approaches. Ostrogradsky approach
is based on the idea that consecutive time derivatives of initial
coordinates form new coordinates. In Schmidt's approach, the
acceleration is defined as a new coordinate instead of the velocity \cite{AnGoMaMa10,EsGu18,Sc94} and Lagrange multipliers do not make an appearance, that is the advantage of the method. Instead, the Lagrangian function is modified by adding a gauge term such that the associated energy function contains additional terms, but no Lagrange multipliers. 
Another important feature about the Schmidt method is that it equally deals with degenerate or nondegerate Lagrangians. A particular case of the method is the accelaration bundle, which arises in Lagrangians on $T^{2}Q$ \cite{Kasper}. As discussed in previous paragraphs, the energy function is a Morse family generating a Lagrangian submanifold. Here, since we add the gauge invariance of a function, we have a different Morse family, and hence, we have a different Hamilton-Jacobi problem. We find this interesting both in theoretical and practical senses.

\bigskip

\noindent \textbf{Organization of the paper.}  In Section 2, we review the fundamentals of higher order bundles and dynamics, Morse families and special symplectic structures, Tulczyjew triples for higher order bundles and implicit higher-order differential equations.
In Section 3, we construct a Hamilton--Jacobi theory for higher order implicit Lagrangian systems. We explain our two main procedures to work with the implicit character of the arising higher order implicit equations: one is the Lagrangian submanifold method or Morse family approach, and the second is the construction of a local vector field by the existence of an additional section that reduces the number of vectors in the implicit submanifold $S$ projecting to a same point of a lower dimensional bundle. For the Morse case we will ellaborate a list of subcases considering the Ostrogradsky method and the Schmidt transformation, comparing both cases and illustrating their usefulness in nondegenerate and degenerate cases.
 Section 4 shows the applications of a implicit Hamilton--Jacobi theory for higher order dynamical systems in the particular case of second order Lagrangians. We will depict such application making use of the Ostragradski approach and the Schmidt-Legendre transform. For second order Lagrangians, we also introduce the setting of the acceleration bundle for the Schmidt-Legendre transform, in order to deal likewise with degenerate or nondegenerate higher order implicit Lagrangians. Two particular examples are a deformed elastic cylindrical beam with fixed ends and the end of a javelin.
  The Ostrogradsky and Schmidt methods will be compared in this same section for nondegenerate cases, as it is the case in which Ostrogradsky applies. As more general models we will depict the second and third order Lagrangians with affine dependence on the acceleration.
Section 5 contains further commentaries on the usefulness of the theory as well as some examples.  

\section{Fundamentals}

Let us consider differential manifolds and standard tensor bundle calculus. It is assumed throughout the text that all structures and mappings are smooth ($C^{\infty}$-class). For very detailed descriptions of fundamentals, we refer to \cite{choquet} and we shall skip to our notation and brief comments on the essentials.

\subsection{Morse families and special symplectic structures} \label{MF-SSS}

\noindent \textbf{Morse families.}
Let $\left(P,\pi, N \right) $ be a fiber bundle.
The vertical bundle $VP$ over $P$ is the space of vertical vectors $U\in TP$
satisfying $T\pi\left( U\right) =0$. The conormal bundle
of $VP$ is defined by
\begin{equation*}
V^{0}P=\left\{ \alpha\in T^{\ast}P:\left\langle \alpha,U\right\rangle
=0,\forall U\in VP\right\} .
\end{equation*}
Let $E$ be a real-valued function on $P$, then the image Im$\left( dE\right)
$ of its exterior derivative is a subspace of $T^{\ast}P $. We say that $E$
is a Morse family (or an energy function) if
\begin{equation}
T_{z}Im\left( dE\right) +T_{z}V^{0}P=TT^{\ast}P,  \label{MorseReq}
\end{equation}
for all $z\in Im\left( dE\right) \cap V^{0}P$. A Morse family defined on $\left( P,\pi ,N\right) $ generates a Lagrangian submanifold of the canonical symplectic structure $\left(
T^{\ast }N,\Omega\right)$ in the following way:
\begin{equation}
S=\left\{ w \in T^{\ast }N:T^{\ast }\pi (w
)=dE\left( z\right) \right\}  \label{LagSub}
\end{equation}%
 In this case, we say that $S$ is generated by the Morse family $E$. Note that, in the definition of $%
S$, there is an intrinsic requirement that $\pi \left( z\right)
=\pi _{N}\left( w \right) $. The inverse of this statement is also true, that is, any Lagrangian submanifold is generated by a Morse family. This is known as the generalized Poincar\'{e} \cite{BeTulcz,jane,LM87,tulcz2,weinsteintubular}. Here, we are presenting the following diagram in order to summarize this discussion. 
 \begin{equation} \label{Morse-pre}
  \xymatrix{
\mathbb{R}& P \ar[d]^{\pi}\ar[l]^{E}& T^*N \ar[d]^{\pi_N}  &S \ar[l]^{i}\\ &
N  \ar@{=}[r]& N &
}
\end{equation}

\bigskip

\noindent \textbf{Local picture for Morse families.}
Assume that $N$ is equipped with local coordinates $(x^a)$, and consider the bundle local coordinates $(x^a,\lambda^\alpha)$ on the total space $P$. In this picture, a function $E$ is called a Morse family if the rank of the matrix%
\begin{equation}
\left( \frac{\partial ^{2}E}{\partial x^a x^b} \quad \frac{%
\partial ^{2}E}{\partial x^a \partial \xi^\alpha }\right)  \label{MorseCon}
\end{equation}%
is maximal. In such a case, the Lagrangian submanifold (\ref{LagSub}) generated by $E$ locally looks like
\begin{equation} \label{MFGen}
S=\left \{\left(x^a,\frac{\partial E}{\partial x^a}(x,\xi)\right )\in T^*N: \frac{\partial E}{\partial \lambda^\alpha}(x,\xi)=0 \right \}.
\end{equation}
See that the dimension of $S$ is half of the dimension of $T^*N$, and that the canonical symplectic two-form $\Omega$ vanishes on $S$.

\bigskip

\noindent \textbf{{Special symplectic structures.}} Let $P$ be a symplectic manifold carrying an exact symplectic two-form $%
\Omega=d\Theta$. Assume also that, $P$ is the total space of a fibre bundle $(P,\pi,M)$. A special symplectic structure is a quintuple $
(P,\pi,M,\Theta,\chi)$ where $%
\chi$ is a fiber preserving symplectic
diffeomorphism from $P$ to the cotangent bundle $T^{\ast}M$. Here, $\chi$ can uniquely be characterized by
\begin{equation} \label{chi-}
\left\langle \chi(p),\pi_{\ast}X(m)\right\rangle =\left\langle
\Theta(p),X(p)\right\rangle
\end{equation}
for a vector field $X$ on $P$, for any point $p$ in $P$ where $\pi(p)=m$. Note that, pairing on the left hand side of (\ref{chi-}) is the natural pairing between the cotangent space $T^*_mM$ and the tangent space $T_mM$. Pairing on the right hand side of (\ref{chi-}) is the one between the cotangent space $T^*_pP$ and the tangent space $T_pP$. We refer \cite{Benenti,LaSnTu75, SnTu73} for further discussions on special symplectic structures. Here is a diagram exhibiting the special symplectic structure.
\begin{equation} \label{sss}
\xymatrix{T^{\ast }M \ar[dr]_{\pi_{M}} &&P
\ar[dl]^{\pi} \ar[ll]_{\chi}
\\&M }
\end{equation}
The two-tuple $(P,\Omega)$ is called as underlying symplectic manifold of the special symplectic structure $%
(P,\pi,M,\Theta,\chi)$.

Let $
(P,\pi,M,\Theta,\chi)$ be a special symplectic structure. Assume also that $S_P$ be a Lagrangian submanifold of $P$. The image $\chi(S_P)$ of $S_P$ is a Lagrangian submanifold of $T^{\ast}M$. By referring to the generalized Poincar\'{e} lemma presented in the previous subsection, we argue that the Lagrangian submanifold $\chi(S_P)$ can locally be generated by a Morse family $E$ on a fiber bundle $(R,\tau,N)$ where $N\subset M$. Accordingly, we are calling the Morse family $E$ as a generator of both $S$ and $S_P$ since they are the same up to $\chi$. The following diagram summarizes this discussion by equipping a Morse family to the special symplectic structure (\ref{sss}).
 \begin{equation} \label{Morse-Gen}
 \xymatrix{
\mathbb{R}& R \ar[d]^{\tau}\ar[l]^{E}& T^*M \ar[dr]_{\pi_M}& &P \ar[ll]_{\chi} \ar[dl]^{\pi}\\ &
N \ar@{(->}[rr]&& M
}
\end{equation}

\subsection{Geometry of higher order bundles}

Given a fibration $(P,\pi,N)$, consider the dimension of $P$ be $p$ and that of $N$ be $n$. Consider a section $s:N\rightarrow P$ and let us denote by Sec$(P)$ the set of all sections on $P$. We say that two sections $s,s'\in$ Sec$(P)$ are $k$-related for $0\leq k\leq \infty$ in a point $x\in N$ if $s(x)=s'(x)$ and for all functions $f:P\rightarrow \mathbb{R}$, the function $f\circ s-f\circ s':N\rightarrow \mathbb{R}$ is flat of order $k$ at $x$, that is, this function and all the derivatives up to order $k$ included are zero at $x$. The equivalence class determined by the $k$-relation is called jet of order $k$ for a section $j^ks(x)$ \cite{book:Saunders89}. The set of all $k$-jets at $x$ is denoted by $J_x^k(P,\pi,N)$. For the union of all of them at any point $x$, we say $J^k(P,\pi,N)$. 
More generally, we can define now at a point a mapping from $N$ to $P$. Consider a map $f:N\rightarrow P$, then the equivalence class determined by the $k$-equivalence is called the $k$-jet of $f$ at $x$. For a representative of the class we use $j^kf(x)$ and the set of $k$-jets is represented by $J_x^k(P,N)$ and again, the union of these at every $x$ will be represented by $J^k(P,N)$. Notice that the manifold $J^k(P,\pi,N)$ is a submanifold of $J^k(P,N)$ and so the above projections admit restrictions to it.

{\it Remark:} Both in the case of sections or mappings it is possible to define jets for local sections or mappings. For it, one works with the germs, recall it is the equivalence class determind by the relation: two section/mappings are related if they have the same value at every point in the intersection of their domains.

The $k$-jet manifold of section/mappings can be fibered in different ways, we have
\begin{align}
&\alpha^k:J^k(P,N)\rightarrow N; \quad \alpha^k(f^k(x))=x,\\
&\beta^k:J^k(P,N)\rightarrow P \quad \alpha^k(f^k(x))=x\\
&\rho^k_r:J^k(P,N)\rightarrow J^r(P,N)\quad \rho_r^k(f^k(x))=f^r(x), \quad r\leq k;  
\end{align}
Here the $\alpha^k$ projection is called the source projection and $\beta^k$ is the target projection.

Now, consider $J^1(J^k(P,N),\alpha^k,N)$ be a manifold of $1$-jets of $\alpha^k:J^k(P,N)\rightarrow N$. The interest of this manifold is that the fibered manifold $J^{k+1}(P,N)$ can be regularly immersed into $J^1(J^k(P,N),\alpha^k,N)$.  Let us simplify the notation by $J^k(P,N)\equiv J^kP$ and so on.

\begin{center}\label{embed}
 \begin{tikzcd}[column sep=huge,row sep=huge]
  &J^{k}P &J^{k+1}P \arrow[l, swap, "\rho^{k+1}_{k}"]\arrow[r, "\psi"]  &J^1(J^{k}P)  \\
 & &N\arrow[lu, "\rho^{k+1}_{k} \circ u"]\arrow[u, "u"]\arrow[ru, swap,"j^1(\rho^{k+1}_{k} \circ u)"]&\\
 \end{tikzcd} 
\end{center}
such that
\begin{equation}
u(x)=J^1 (\rho^{k+1}_{k} \circ u)(x)
\end{equation}
for a function $u:N\rightarrow J^{k+1} P$. Note that $J^{k+1}P\subset J^1(J^k P)$.
We can set local coordinates for jets, the jet manifold $J^k(P,N)$ has an atlas when it is modeled in the space $\mathbb{R}^p\times \mathbb{R}^n\times J^k(\mathbb{R}^p,\mathbb{R}^n)$, locally we may think $J^k(P,N)$ as $J^k(\mathbb{R}^p, \mathbb{R}^n)$. Locally, $(x_a,\xi^A, \xi^A_{a(r)})$ with $1\leq r\leq k$ is the coordinate representation of a point of $J^k(P,N)$.

Now, a particular type of jet manifold is the tangent bundle of higher order. Consider $Q$ a configuration space of dimension $n$, $TQ$ is the tangent bundle and $T^{*}Q$ is the cotangent bundle or phase space for a dynamical system. The $k$-order tangent bundle can be identified with $k$ order jets in the following way

\begin{equation}
T^kQ=J_0^k(\mathbb{R},Q)=J^k(\mathbb{R}\times Q,\pi_1,  \mathbb{R})
\end{equation}
($T^kQ$ is a submanifold of $J^k(\mathbb{R},Q)$). One has the same type of fibrations as for the jets above. In fact, if $r\leq k$, we have the canonical projection $\rho_r^k:T^kQ\rightarrow T^rQ$, given by $\rho_r^k(\sigma^k(0))=\sigma^r(0)$, and the target projection is $\beta^k:T^kQ\rightarrow Q$, given by $\beta^k(\sigma^k(0))= \sigma(0)$. One has obviously $\rho^k_0=\beta^k$, where $T^0Q$ is identified canonically with $Q$.

To describe the local coordinates in $T^kQ$, let $\left(U,\varphi\right)$ be a local chart in $Q$,
with $\varphi = (\varphi^A)$, $1\leqslant A \leqslant n$,
and $\varphi \colon \R \to Q$ is a curve in $Q$ such that $\varphi(0) \in U$;
by writing $\phi^A = \varphi^A \circ \phi$, the $k$-jet $\tilde{\phi}^k(0)$ is uniquely represented in
 $\left(\beta^k\right)^{-1}(U) = T^kU$ by

\begin{equation}  \label{coorT2}
( {q}^A _{(0)}, {q}^A _{(1)}, {q}^A _{(2)},..., {q}^A _{(k)}):=( {q}^A ,{\dot q}^A, \ddot{q}^A,..., {q} ^{(k)A})
\end{equation}
where
$$
q^A = \phi^A(0) \quad ; \quad
q_{(i)}^A = \restric{\frac{d^{(i)}\phi^A}{dt^{(i)}}}{t=0} \, \quad i=1,\dots,k
$$
in the open set $\left(\beta^k\right)^{-1}(U) \subseteq \Tan^kQ$.
The local expression of the canonical projections $\beta^k$ and $\rho^k_r$ are
$$
\rho^k_r\left(q_{(0)}^A,q_{(1)}^A,\ldots,q_{(k)}^A\right) = \left(q_{(0)}^A,q_{(1)}^A,\ldots,q_{(r)}^A\right), \qquad
\beta^k\left(q_{(0)}^A,q_{(1)}^A,\ldots,q_{(k)}^A\right) = \left(q_{(0)}^A\right) \, .
$$
Hence, local coordinates in the open set $\left(\beta^k\right)^{-1}(U) \subseteq \Tan^kQ$
adapted to the $\rho^k_r$-bundle structure are $$\left(q_{(0)}^A,\ldots,q_{(r)}^A,q_{(r+1)}^A,\ldots,q_{(k)}^A\right),$$
and a section $s \in \Gamma(\rho^k_r)$ is locally given in this open set by
$$s(q_{(0)}^A,\ldots,q_{(r)}^A) = \left( q_{(0)}^A,\ldots,q_{(r)}^A,s_{(r+1)}^A,\ldots,s_{(k)}^A \right),$$
where $s_{(j)}^A$ (with $r+1 \leqslant j \leqslant k$) are local functions.
This approach is very useful to work on the tangent bundle $TT^{k-1}Q$. Accordingly, we denote the induced coordinates on $TT^{k-1}Q$ as $$( {q} _{(\kappa)}^A;{\dot{q}}_{(\kappa)}^A)=( {q}^A _{(0)}, {q}^A _{(1)},..., {q}^A _{(k)}
 ;
 {\dot q}_{(0)}^A,{\dot q}_{(1)}^A,...,{\dot q}_{(k)}^A)\in TT^{k-1}Q ,$$
 where $\kappa$ runs from $0$ to $k-1$. 

\subsection{Tulczyjew triples for higher order bundles}

Consider the natural embedding of $T^kQ$ into the iterated tangent bundle $TT^{k-1}Q$ of $T^{k-1}Q$. This is locally given by
\begin{equation} \label{emb}
\iota:T^kQ\mapsto TT^{k-1}Q: (q_{(0)},{q}_{(1)},\dots, q_{(k)}) \mapsto (q_{(0)},{q}_{(1)},\dots q_{(k-1)}; {q}_{(1)},\dots q_{(k-1)}, q_{(k)}),
\end{equation}
see \cite{book:DeLeon_Rodrigues85}. Here, the induced coordinates on $TT^{k-1}Q$ are assumed to be 
\begin{equation}\label{indcoor}
( {q} _{(\kappa)}^A;{\dot{q}}_{(\kappa)}^A)=( {q}^A _{(0)}, {q}^A _{(1)},..., {q}^A _{({k-1})}
 ;
 {\dot q}_{(0)}^A,{\dot q}_{(1)}^A,...,{\dot q}_{({k-1})}^A)\in TT^{k-1}Q ,
 \end{equation}
 where $\kappa$ runs from $0$ to ${k-1}$. For future reference, let us record here the particular case $s=2$ that is
\begin{equation} \label{sot}
T^2Q\longrightarrow TTQ: ( q_{(0)}^A,q_{(1)}^A,q_{(2)}^A) \longrightarrow
(  q_{(0)}^A,q_{(1)}^A;q_{(1)}^A,q_{(2)}^A).
\end{equation}
This embedding will enable us to study the dynamics on the higher order bundles in the framework of Tulczyjew triples. 

Recall first the first order Tulczyjew triple \cite{GrabGrab,Ibort,MMT,tulcz1, tulcz2,tulcz3,ZG}. By replacing the configuration manifold $Q$ in the clasical first order Tulczyjew triple by the $(k-1)$-th order tangent bundle $T^{k-1}Q$,
we draw the following generalized Tulczyjew's triple to higher order dynamics \cite{EsGu18}
\begin{equation}
\xymatrix{T^{\ast }TT^{k-1}Q \ar[dr]_{\pi_{TT^{k-1}Q}}&&TT^{\ast }T^{k-1}Q\ar[dl]^{T\pi_{T^{k-1}Q}}
\ar[rr]^{\Omega_{T^{k-1}Q}^{\flat}} \ar[dr]_{\tau_{T^{\ast }T^{k-1}Q}}
\ar[ll]_{\Xi_{T^{k-1}Q}}&&T^{\ast }T^{\ast }T^{k-1}Q\ar[dl]^{\quad \pi _{T^{\ast}T^{k-1}Q}}
\\&TT^{k-1}Q&&T^{\ast}T^{k-1}Q }. \label{TT}
\end{equation}

\noindent
that shows the passing from the tangent to a higher order bundle to its cotangent.
Here, $\pi_{TT^{k-1}Q}$ is the cotangent bundle projection, $T\pi_{T^{k-1}Q}$ is the tangent lift of $\pi_{T^{k-1}Q}$, $\tau_{T^{\ast }T^{k-1}Q}$ is the tangent bundle projection, and $\pi _{T^{\ast}T^{k-1}Q}$ is the cotangent bundle projection. 

Since $T^*T^{k-1}Q$ is a cotangent bundle, the pair $(T^*T^{k-1}Q, \Omega_{T^{k-1}Q})$ is a symplectic manifold with the canonical symplectic two-form $\Omega_{T^{k-1}Q}=d\Theta_{T^{k-1}Q}$.
On $T^*T^{k-1}Q$, the Darboux coordinates are 
$$( {q} _{(\kappa)}^A;{p}^{(\kappa)}_A)=( {q}^A _{(0)}, {q}^A _{(1)},...,{q}^A _{(k-1)},{p}^{(0)}_A,{p}^{(1)}_A,
...,{p}^{(k-1)}_A)\in T^*T^{k-1}Q,$$ so we write the canonical two-form as
\begin{equation}
\Omega_{T^{k-1}Q}= \sum_{\kappa=0}^{k-1}  d{p}^{(\kappa)}_A \wedge d {q}_{(\kappa)}^A.
\end{equation}
On $TT^*T^{k-1}Q$, introduce the following local coordinate system
\begin{equation} \label{CoordTT*TQ}
( {q} _{(\kappa)}^A,{p}^{(\kappa)}_A,
{\dot{q}}_{(\kappa)}^A,{\dot{p}}^{(\kappa)}_A)\in TT^*T^{k-1}Q
\end{equation}
where $\kappa$ runs from $0$ to ${k-1}$. The pair $\left(TT^*T^{k-1}Q, \Omega^T_{T^{k-1}Q}\right)$ is a symplectic manifold with lifted symplectic two-form. In terms of the coordinates, $\Omega^T_{T^{k-1}Q}$ can be written as
\begin{equation}\label{omegat}
\Omega^T_{T^{k-1}Q}= \sum_{\kappa=0}^{k-1} d{\dot p}^{(\kappa)}_A\wedge d  {q} _{(\kappa)}^A
+
\sum_{\kappa=0}^{k-1} d{p}^{(\kappa)}_A\wedge d{\dot q}_{(\kappa)}^A.
\end{equation}
Then, we define the adapted symplectic diffeomorphism $\Xi_{T^{k-1}Q}$ and $\Omega^\flat_{T^{k-1}Q}$ from the symplectic diffeomorphism $\Xi_Q$ and $\Omega^\flat_{Q}$ in the first order Tulczyjew triple \eqref{TT}. Accordingly, they are computed as
\begin{align}\label{morph}
 \Xi_{T^{k-1}Q}( q_{(\kappa)},p^{(\kappa)},
{\dot{q}}_{(\kappa)},{\dot{p}}^{(\kappa)})
&=
( q_{(\kappa)},
{\dot{q}}_{(\kappa)},{\dot{p}}^{(\kappa)},{p}^{(\kappa)}),
\\
\Omega^\flat_{T^{k-1}Q}( q_{(\kappa)},{p}^{(\kappa)},
{\dot{q}}_{(\kappa)},{\dot{p}}^{(\kappa)}) \nonumber
&=(q_{(\kappa)},{p}^{(\kappa)},{\dot{p}}^{(\kappa)},-
\dot{q}_{(\kappa)}), \quad \kappa=0,\dots,k-1.
\end{align}
We remark here that both the left and the right wings of the higher order Tulczyjew's triple are special symplectic structures, and the triple is merging them to enable a Legendre transformation for the singular or/and constrained higher order dynamical systems. 

\subsection{Explicit higher order differential equations}

Consider the bundle projection $\pi_1 \colon \R \times Q \to \R$ onto the first factor. If $\phi \colon \R \to Q$ is a curve in $Q$,
the {\rm canonical lifting} of $\phi$ to $\Tan^kQ$
is the curve $j^k\phi\colon \Real\to\Tan^kQ$.
We consider the module of vector fields $\vf(\pi^r_s)$ along the projection $\pi^r_s:J^r\pi\rightarrow J^s\pi$.
The $k$th holonomic lift of  $\ds X = X_o \derpar{}{t}\in \vf(\R)$  is given by 
$$
j^kX = X_o\left( \derpar{}{t} + \sum_{i=0}^{k} q_{(i+1)}^A\derpar{}{q_{(i)}^A}\right) \, .
$$

Using the identification $J^k\pi \cong \R \times \Tan^kQ$ and denoting by
$\pi_2 \colon \R \times Q \to Q$ the natural projection onto the second factor, and
all the induced projections in higher order jet bundles,
we have the following diagram.
$$
\xymatrix{
\R \times \Tan^{k+1}Q \ar[dd]_{\pi^{k+1}_{k}} \ar[rr]^-{\pi_2^{k+1}} & \ & \Tan^{k+1}Q \ar[dd]^{\rho^{k+1}_{k}} \\
\ & \ & \ \\
\R \times \Tan^{k}Q \ar[rr]^-{\pi_2^k} & \ & \Tan^{k}Q
}
$$


\begin{definition}
A curve $\psi\colon \R \to \Tan^{k}Q$ is {\rm holonomic of type $r$}, $1 \leqslant r \leqslant k$,
if $j^{k-r+1}\phi = \rho^{k}_{k-r+1} \circ \psi$,
where $\phi = \beta^{k} \circ \psi, \phi \colon \R \to Q$;
that is, the curve $\psi$ is the lifting of a curve in $Q$ up to $\Tan^{k-r+1}Q$.
$$
\xymatrix{
\ & \ & \Tan^kQ \ar[d]^{\rho^k_{k-r+1}}  \\
\R \ar[urr]^{\psi} \ar[ddrr]_{\phi = \beta^k\circ\psi}
\ar[rr]^{\rho^k_{k-r+1}\circ \psi} \ar[drr]^{j^{k-r+1}\phi}
& \ & \Tan^{k-r+1}Q \ar[d]^{{\rm Id}} \\
\ & \ & \Tan^{k-r+1}Q \ar[d]^{\beta^{k-r+1}} \\
\ & \ & Q
}
$$
In particular, a curve $\psi$ is {\rm holonomic of type 1}
if $j^{k}\phi = j^k_0\phi$, with $\phi = \beta^{k} \circ \psi$.
Throughout this paper, holonomic curves of type $1$ are simply called {\rm holonomic}.
\end{definition}

\begin{definition}
A vector field $X \in \vf(\Tan^kQ)$ is a {semispray of type $r$}, $1 \leqslant r \leqslant k$,
if every integral curve $\psi$ of $X$ is holonomic of type $r$.
\end{definition}

The local expression of a semispray of type $r$ is
\begin{equation}\label{semispray}
X = q_{(1)}^A\derpar{}{q_{(0)}^A} + q_{(2)}^A\derpar{}{q_{(1)}^A} + \ldots + q_{(k-r+1)}^A\derpar{}{q_{(k-r)}^A} +
F_{k-r+1}^A\derpar{}{q_{(k-r+1)}^A} + \ldots + F_k^A\derpar{}{q_{(k)}^A} \, .
\end{equation}

where $F_{k-r+1}^A,\dots,F_k^A$ are functions of $q_{(i)}, i=1,\dots,k$.

Observe that semisprays of type $1$ in
$\Tan^kQ$ are the analogue to holonomic vector fields in first order mechanics.
Their local expressions are
\begin{equation}\label{semispray1}
X = q_{(1)}^A\derpar{}{q_{(0)}^A} + q_{(2)}^A\derpar{}{q_{(1)}^A} + \ldots + q_{(k)}^A\derpar{}{q_{(k-1)}^A} + F^A\derpar{}{q_{(k)}^A} \, .
\end{equation}

If $X\in\vf(\Tan^kQ)$ is a semispray of type $r$, a
curve $\phi \colon \R \to Q$ is said to be a {\sl path} or {\sl
solution} of $X$ if $j^k\phi$ is an integral curve of $X$; that is,
$\widetilde{j^k\phi} = X \circ j^k\phi$, where $\widetilde{j^k\phi}$
denotes the canonical lifting of $j^k\phi$ from $\Tan^kQ$ to
$\Tan \Tan^kQ$.

\subsection{Implicit higher order differential equations}

Consider a $k$-th order system 
$$\Phi^l(q,\dot{q},\ddot{q},\dots, q^{(k)})=0, \qquad l=1,\dots, r$$
of differential equations defined by $r$ number equations on $T^kQ$. Geometrically, the functions $\Phi^l$ define a submanifold $S$ of $T^kQ$. Using the induced coordinates on the higher order tangent bundle, this submanifold is given locally by
\begin{equation}
S=\{\mathbf{q}:=(q,\dot{q},\ddot{q},\dots, q^{(k)})| \quad \Phi^l(\mathbf{q})=0 \}.
\end{equation}
A differentiable curve $\phi$ on $Q$ whose canonical $k$-lifting is a curve $\psi=j^k_0\phi$ on $T^kQ$ is a solution of $S\subset T^kQ$ if the lifted curve lies in $S$. 

The submanifold $S$ can be understood as a first order differential equation defined on $TT^{k-1}Q$ as well. To this end we first consider the natural embedding of $T^kQ$ into the iterated tangent bundle $TT^{k-1}Q$ of $T^{k-1}Q$. This is locally described as in \eqref{emb}. For $s=2$, recall \eqref{sot}.
Using the mapping in (\ref{emb}), image $\iota(S)$ of $S$ is a submanifold of $TT^{k-1}Q$.   
The differential equation $S$ is called explicit if there exists a vector field $X$ on $T^{k-1}Q$ such that $Im(X)$ is $\iota(S)$. Otherwise, $S$ is called an implicit differential equation.

Looking for a Lagrangian function generating a differential equation is the inverse problem of calculus of variations. See, for example, \cite{MoFeVeMaRu90} for a geometric approach to this problem for the case of $s=1$. For the fourth order explicit systems, in \cite{Fe96}, some conditions are proposed for the existence and uniqueness of a Lagrangian function. In this work, we assume that there exists already a Lagrangian function generating the dynamics.

\section{The Hamilton-Jacobi problem for higher order implicit systems}

The submanifold $S$ defined by a higher-order regular Lagrangian in $TT^{*}T^{k-1}Q$ projects via $T\pi_{{T^{k-1}Q}}$ on the whole $T^{*}T^{k-1}Q$. In the singular case, this projection is only a part of $T^{*}T^{k-1}Q$. If we would like to construct a Hamilton--Jacobi theory in this setting, there must be a way in which we obtain a Lagrangian submanifold of $TT^{*}T^{k-1}Q$. To find a solution, we need to find a section $\gamma:T^{*}T^{k-1}Q\rightarrow TT^{*}T^{k-1}Q$. Nonetheless, starting from an implicit differential equation on $TT^{*}T^{k-1}Q$, by the projection $T\pi_{T^{k-1}Q}$, we arrive at a submanifold in $TT^{k-1}Q$. Hence, we need to make use of Tulcyjew's triple (\ref{TT}) to pass from the Lagrangian to Hamiltonian pictures $TT^{k-1}Q$ and $T^{*}T^{k-1}Q$ through some morphisms. In this section we develop a geometric Hamilton-Jacobi theory for higher order implicit differential equations using two different approaches.
 
The first method consists of a theory which does refer to vector fields, that we will refer to as the Morse family method. The second is based on the construction of a local vector field defined on the image of a section, but not defined globally on the phase space. In this case, the definitions above apply locally, and the philosophy of the Hamilton--Jacobi approach can match the explanation right abovementioned. Let us then start first with the method which is not so related to the usual definitions and that implements as a novelty the use of Lagrangian submanifolds generated by a Morse function. Hence, we start with our so-called Morse family method.

Notice also that we will rely on the Ostrogradsky approach \eqref{CoordTT*TQ} in this subsection, but there is an alternative, the Schmidt approach that we will introduce in the next section.

 \subsection{The Morse family method - General approach}
 
Let us start with the first method. We start with a Lagrangian submanifold $S$ of the symplectic manifold $TT^{*}T^{k-1}Q$ equipped with the symplectic two-form $\Omega^T_{T^{k-1}Q}$ exhibited in (\ref{omegat}). If it is a horizontal Lagrangian submanifold then it is possible to find a Hamiltonian vector field on $T^{*}T^{k-1}Q$ whose image is exactly the submanifold itself \cite{SnTu73}. This is corresponding to an explicit dynamical system. If the Lagrangian submanifold fails to be horizontal then there is no Hamiltonian vector field generating the Lagrangian submanifold. In this case, dynamical equations governing the dynamical system can only be written in an implicit differential equation form. We now propose a Hamilton-Jacobi formalism valid both for the explicit and implicit systems. 

Consider a Lagrangian submanifold $S$ of $TT^{*}T^{k-1}Q$. If it is projectable, by projecting $S$ via the mapping $\tau_{T^{*}T^{k-1}Q} $, we reach a submanifold of $T^*T^{k-1} Q$. On the other hand, if we project $S$ with $T{\pi_{T^{k-1}Q}}$, we reach a submanifold of $TT^{k-1}Q$, where we have a first order implicit differential equation. Accordingly, we can iteratively project these resulting bundles: 
from $T^{*}T^{k-1}Q$ to $Q$. Let us summarize this in the following diagram.
\begin{center}
 \begin{tikzcd}[column sep=tiny,row sep=huge]
  & &  &S\arrow[r, hook, "i"] &TT^*T^{k-1}Q \arrow[ld, "\tau_{T^*T^{k-1}Q}"] \arrow[rd, "T{\pi_{T^{k-1}Q}}"] & & \\
  & & &T^*T^{k-1}Q\arrow[dr,"\pi_{T^{k-1}Q}"]& &TT^{k-1}Q\arrow[dl,"\tau_{T^{k-1}Q}"]\arrow[dr,"T\tau^{k-1}_Q"]  &&S^{\gamma}\arrow[ll, hook', "i"']\\
& && &T^{k-1}Q\arrow[d,"\tau_Q^{k-1}"]\arrow[ul,"\gamma",bend left=20]
 & &TQ\arrow[dll,"\tau_Q"'] &\\
 & & & &Q & &\mathbb{R}\arrow[ll,"\phi"']\arrow[u,"j^1\phi"']  &\\
 \end{tikzcd}
\end{center}
where $\tau_Q^{k-1}:T^{k-1}Q\rightarrow Q$.

Notice that if $S$ is integrable, then 
$T{\pi_{T^{k-1}Q}}(S)$ is integrable too. We see this by considering the projection $T{\pi_{T^{k-1}Q}}(V)$ of an element $V\in S$. Note that, if $\varphi$ is a curve lying in  $T^*T^{k-1}Q$ and it is tangent to $V\in S$, then $ \pi_{T^{k-1}Q}\circ \varphi$ is curve on $T^{k-1}Q$ that is tangent to $T{\pi_{T^{k-1}Q}}(V)$. This shows that the projections of the solutions of $S$ are solutions of $T{\pi_{T^{k-1}Q}}(S)$.
The inverse question is precisely the basis of a Hamilton--Jacobi theory, i.e., if starting from the solutions of $T{\pi_{T^{k-1}Q}}(S)$
we are able to construct solutions of $S$, that is to lift the solutions on $TT^{k-1}Q$ to the iterated bundle $TT^{*}T^{k-1}Q$.

Notice that $S$ may not be projectable, that means that $S$ is only projectable when it is restricted to the image space of a differential one-form $\gamma$ on $T^{k-1}Q$. We denote the restriction of $S$ to the image space of a one-form $\gamma$ as follows $S|_{Im(\gamma)}$.  
\noindent
For this procedure, we need to introduce a section $\gamma:T^{k-1}Q\rightarrow T^*T^{k-1}Q$ such that for a solution $\psi:\mathbb{R}\rightarrow T^{k-1}Q$ of $S^{\gamma}=T\pi_{T^{k-1}Q} (S)$, we have that $\gamma\circ\psi:\mathbb{R}\rightarrow T^*T^{k-1}Q$ is a solution of $S$. We say that $S$ and $S^\gamma$ are $\gamma-$related. In accordance to the usual Hamilton--Jacobi theory \cite{BMMP-10,HJteam-K,HJTeam2,leones2}, recall \eqref{HJT}, we have
\begin{equation}\label{Xg}
  \xymatrix{ T^{*}T^{k-1}Q
\ar[dd]^{\pi_{T^{k-1}Q}} \ar[rrr]^{\widetilde{j^{k-1}}\circ \gamma\circ \phi}&   & &TT^{*}T^{k-1}Q\ar[dd]^{T\pi_{T^{k-1}Q}} &S\ar@{^{(}->}[l]\\
  &  & & &\\
T^{k-1}Q\ar@/^2pc/[uu]^{\gamma}\ar[rrr]^{\widetilde{j^{k-1}}\circ \phi} &  & & TT^{k-1}Q &S^{\gamma}\ar@{^{(}->}[l]}
\end{equation}
In this case, since $S$ and $S^{\gamma}$ are implicit, we do not have a vector field.
Nonetheless, as we summarized in Section \eqref{MF-SSS}, for every Lagrangian submanifold $S$ in $TT^*T^{k-1}Q$, there exists a Morse family $E$ defined over a smooth bundle structure 
$(R,\tau,T^*T^{k-1}Q)$ generating $S$. Let us recall in a diagram the Lagrangian submanifold that is generated and the Lagrangian submanifold we need for a Hamilton--Jacobi theory:
\begin{equation}
\xymatrix{ 
&&D\ar@{^{(}->}[d]&&S\ar@{^{(}->}[d]
\\
\mathbb{R}& R \ar[d]^{\tau}\ar[l]_{E}& T^*T^{*}T^{k-1}Q\ar[d]_{\pi_{T^{*}T^{k-1}Q}}& &TT^{*}T^{k-1}Q \ar[ll]_{\Omega_{T^{k-1}}^{\flat}} \ar[dll]^{\quad \tau_{{T^{*}T^{k-1}Q}}}\\ &
T^{*}T^{k-1}Q \ar@{=}[r]& T^{*}T^{k-1}Q
} 
\end{equation}
where the triangle is the special symplectic structure presented as the right wing of the Tulczyjew triple (\ref{TT}). Here, $D$ is the image of $S$ under the musical mapping $\Omega_{T^{k-1}Q}^{\flat}$ hence a Lagrangian submanifold of $T^*T^{*}T^{k-1}Q$. Assume the local coordinates $(q^A_{(\kappa)},p_A^{(\kappa)},\lambda^\alpha)$ on the fiber bundle $R$. Here $(q^A_{(\kappa)},p_A^{(\kappa)})$ is the Darboux' coordinates on $T^{*}T^{k-1}Q$ since $\kappa$ runs from $0$ to $k-1$. The Lagrangian submanifold $S$, generated by the Morse family $E=E(q_{(\kappa)},p^{(\kappa)},\lambda)$, can be written as 
\begin{equation} \label{EandF}
D=\left \{\left(q^A_{(\kappa)},p_A^{(\kappa)};\frac{\partial E}{\partial q^A_{(\kappa)}},\frac{\partial E}{\partial p_A^{(\kappa)}}\right )\in T^{*}T^{*}T^{k-1}Q:\frac{\partial E}{\partial \lambda^\alpha}=0\right\}.
\end{equation}
The isomorphic image of $D$ is the Lagrangian submanifold describing the dynamics and computed to be
\begin{equation} \label{EandF2}
S=\left \{\left(q^A_{(\kappa)},p_A^{(\kappa)};\frac{\partial E}{\partial p_A^{(\kappa)}},-\frac{\partial E}{\partial q^A_{(\kappa)}}\right )\in TT^*T^{k-1}Q:\frac{\partial E}{\partial \lambda^\alpha}=0\right\}.
\end{equation}
The Lagrangian submanifold $S$ generates the following systems of implicit differential equations
 \begin{equation}
\dot{q}^A_{(\kappa)}=\frac{\partial E}{\partial p_A^{(\kappa)}}, \qquad \dot{p}_A^{(\kappa)}=-\frac{\partial E}{\partial q^A_{(\kappa)}}, \qquad \frac{\partial E}{\partial \lambda^\alpha}=0.
\end{equation}

We introduce a closed one-form $\gamma$ on $T^{k-1}Q$ with local picture $$\gamma(q_{(\kappa)})=\gamma^{(\kappa)}_Adq_{(\kappa)}^A,$$
where $\gamma^{(\kappa)}_A$ are real valued functions on $T^{k-1}Q$. 
See that, Im$(\gamma)$ is a Lagrangian submanifold of $ T^*T^{k-1}Q$, so that there is an inclusion $\iota:Im(\gamma)\mapsto T^*T^{k-1}Q$. We use the inclusion to pull the bundle $(R,\tau,T^*T^{k-1}Q)$ back over Im$(\gamma)$. By this, one arrives at a fiber bundle $(\iota^*R,\iota^*\tau,Im(\gamma))$. 
\begin{equation}\label{pbb-F}
  \xymatrix{
\iota^*(R) \ar[rr] ^{\varepsilon} \ar[dd]^{\iota^*\tau}&& R \ar[dd]^{\tau}\\ \\
Im(\gamma) \ar [rr]_{\iota}&& T^*T^{k-1}Q
}
\end{equation}
Here, the total space the pull-back bundle is
$$\iota^*(R)=\left \{(\gamma(q_{(\kappa)}),z)\in Im(\gamma)\times R : \tau(z)\in Im(\gamma)\right \}$$
with $\varepsilon$ is the corresponding inclusion. Although restriction of the Morse family on $\iota^*(R)$ should formally be written as $E\circ \epsilon$, we will abuse notation using $E$. The submanifold generated by $E=E(q_{(\kappa)},\gamma^{(\kappa)},\lambda)$ is given by
\begin{equation} \label{EandF-gamma}
S\vert_{Im({\gamma})}=\left \{\left(q^A_{(\kappa)},\gamma_A^{(\kappa)};\frac{\partial E}{\partial q_A^{(\kappa)}},-\frac{\partial E}{\partial q^A_{(\kappa)}}\right )\in TT^*T^{k-1}Q:\frac{\partial E}{\partial \lambda^\alpha}=0\right\}.
\end{equation}
 Note that, if the Lagrangian submanifold $S$ was explicit, and would be understood as the image of a Hamiltonian vector field $X_H$, then $S\vert_{Im({\gamma})}$ reduces to the image space of the composition $X_H \circ \gamma$.

The submanifold $S\vert_{Im({\gamma})}$ exhibited in (\ref{EandF-gamma}) does not depend on the momentum variables. This enables us to project it to a submanifold $S^{\gamma}$ of $TQ$ by the tangent mapping $T\pi_Q$ as follows
\begin{equation} \label{E-gamma}
S^{\gamma}=T\pi_{T^{k-1}Q}\circ S\vert_{Im({\gamma})}=\left \{\left(q^A_{(\kappa)},\frac{\partial E}{\partial q_A^{(\kappa)}}(q_{(\kappa)},\gamma^{(\kappa)},\lambda)\right )\in TT^{k-1}Q:\frac{\partial E}{\partial \lambda^\alpha }=0\right\}.
\end{equation}
Note that the submanifold $S^{\gamma}$ defines an implicit differential equation on $T^{k-1}Q$. We state the generalization of the Hamilton-Jacobi theorem (\ref{HJT}) as follows. 
\begin{theorem}[Higher order implicit HJ theorem] \label{nHJT}
The following conditions are equivalent for a closed one-form $\gamma$ that is a solution of the implicit higher order Hamilton--Jacobi problem: 
\begin{enumerate}
\item The Lagrangian submanifold $S$ in (\ref{EandF}) and the submanifold $S^{\gamma}$ in (\ref{E-gamma}) are $\gamma$-related, that is
 $$T\gamma(S^{\gamma})=S\vert_{Im({\gamma})}$$

\item The Morse family $E$ that generates the submanifold $S$ fulfills the equation 

\begin{equation}\label{IHHJE}
dE(q_{(\kappa)},\gamma^{(\kappa)},\lambda)=0,
\end{equation}
\end{enumerate}
\end{theorem}

\begin{proof} 
The one-form $\gamma=\gamma^{(\kappa)}_Adq^A_{(\kappa)}$ is closed, that is, ${\partial \gamma^{(\kappa)}_A}/{\partial q^B_{(\hat{\kappa})}}={\partial \gamma^{(\hat{\kappa})}_B}/{\partial q^A_{(\kappa)}}$. The first assertion in Theorem \ref{nHJT} can be written locally as
\begin{equation} \label{lift}
\delta^\kappa_{\hat{\kappa}}\frac{\partial \gamma^{(\kappa)}_A}{\partial q^A_{(\hat{\kappa})}}\frac{\partial E}{\partial p^{(\kappa)}_A}+\frac{\partial E}{\partial q^A_{(\kappa)}}=0, 
\end{equation}
for $A=1,\dots,n$, and $\hat{\kappa},\kappa=1,\dots,k-1$ and with the condition $\partial E/ \partial \lambda^\alpha =0$. Let us now compute
\begin{eqnarray}\label{pfM}
dE(q_{(\kappa)},\gamma^{(\kappa)},\lambda)&=&\frac{\partial E}{\partial q_{(\kappa)}^A}dq_{(\kappa)}^A+\frac{\partial E}{\partial p^{(\kappa)}_A}\gamma^{(\kappa,\hat{\kappa})}_Adq_{(\hat{\kappa})}^A +\frac{\partial E}{\partial \lambda^\alpha}d\lambda^\alpha
 \end{eqnarray}
 Note that, after the substitution of (\ref{lift}) into (\ref{pfM}) and by employing the closure of the one-form, we conclude that the exterior derivative
 of $E$ vanishes when $p^{(\kappa)}=\gamma^{(\kappa)}$. 
 \end{proof}

 \subsubsection{The Morse family method - Ostrogradsky Momenta}

Now, we come to the problem of deciding the total space $R$ of the bundle $T^{*}T^{k-1}Q$. We are proposing two alternative ways for this. In this case our interest is focused in the Lagrangian submanifolds generated by a Lagrangian function. Accordingly, consider a Lagrangian function depending on higher order differential terms on the higher order tangent bundle $T^kQ$ of the configuration space $Q$. If $Q$ is an $n$-dimensional manifold with a local chart $(q_{(0)}^A)$, then $T^kQ$ is a $(k+1)\times n$-dimensional manifold with the induced local chart $(q_{(0)}^A,q_{(1)}^A,\dots, q_{(k)}^A)$. Now, consider the Whitney product 
\begin{equation} \label{W-1}
W=T^kQ\times_{T^{k-1}Q} T^*T^{k-1}Q
\end{equation} 
equipped with the local coordinates 
\begin{equation}
(q_{(\kappa)}^A,q_{(k)}^A,p^{(\kappa)}_A)=(q_{(0)}^A,\dots,q_{(k-1)}^A, q_{(k)}^A,p^{(0)}_A,\dots p^{(k-1)}_A)
\end{equation} 
of the higher order tangent bundle $T^kQ$ and the iterated cotangent bundle $T^*T^{k-1}Q$ fibered over $T^{k-1}Q$. Here, we have assumed the canonical coordinates $(q_{(\kappa)}^A,p^{(\kappa)}_A)$ on $T^*T^{k-1}Q$ where $\kappa$ runs from $0$ to $k-1$. Note that, we can realize this Whitney product as the total space of the smooth fiber bundle
\begin{equation}\label{Mso}
  \tau: T^kQ\times_{T^{k-1}Q} T^*T^{k-1}Q\mapsto T^*T^{k-1}Q: (q_{(\kappa)}^A,q_{(k)}^A,p^{(\kappa)}_A) \mapsto (q_{(\kappa)}^A,p^{(\kappa)}_A),
\end{equation}
where the base is $T^*T^{k-1}Q$. In this fibration the fibers are given by $(q_{(k)}^A)$ and they are $n$-dimensional.

For a given higher order Lagrangian $L=L(q_{(0)}^A,\dots, q_{(k)}^A)$, the corresponding energy function $E$ is defined on the Whitney product $T^kQ\times_{TQ} T^*T^{k-1}Q$ and explicitly given by
\begin{equation}\label{Morse-Ost}
E(q_{(\kappa)}^A,q_{(k)}^A,p^{(\kappa)}_A)=p^{(0)}_Aq_{(1)}^A+p^{(1)}_Aq_{(2)}^A+\dots + p^{(k-1)}_Aq_{(k)}^A-L.
\end{equation}
It is immediate to see that $E$ is a Morse family and that it generates a Lagrangian submanifold of the cotangent bundle $T^{*}T^*TQ$,
as it was mentioned in the theory on Morse families \eqref{MF-SSS}. Diagrammatically, we replace the total space $R$ with the Whitney product in \eqref{W-1} with the projection \eqref{Mso}. Hence, the Lagrangian submanifold $D$ in (\ref{EandF}) takes the particular form 
\begin{equation} \label{D-1--}
\left({q}_{(\kappa)}^A,{p^{(\kappa)}_A};-\frac{\partial L}{\partial {q}_{(0)}^A}, p^{(0)}_A-\frac{\partial L}{\partial q_{(1)}^A}, \dots, 
p^{(k-2)}_A-\frac{\partial L}{\partial q_{(k-1)}^A},{q}_{(1)}^A,\dots q_{(k)}^A\right)
\end{equation}  
equipped with constraints 
\begin{equation} \label{D-1-}
p^{(k-1)}_A- \frac{\partial L}{\partial  q_{(k)}^A}={0}.
\end{equation}
See that the Lagrangian submanifold exhibited in (\ref{D-1--}) and (\ref{D-1-}) is in $T^*T^*T^{k-1}Q$ where $q_{(k)}^A$ are auxiliary variables presenting the implicit character of the system. In this framework, Ostrogradsky momenta are given by
\begin{equation} \label{OstMom}
{p}^{(\kappa)}=\sum_{j=\kappa}^{k-1}\left(-\frac{d}{dt}\right)^{j-\kappa}\left(\frac{\partial L}{%
\partial q^{(\kappa+1)}}\right)
\end{equation}
where $\kappa$ runs from $0$ to $k-1$.

\medskip

\noindent \textbf{{Legendre transformation by means of Tulczyjew's triple.}} 
Using the right wing of the Tulczyjew' triple (\ref{TT}), and referring directly to the musical isomorphism $\Omega_{T^{k-1}Q}^T$ in (\ref{morph}), we map the Lagrangian submanifold in (\ref{D-1--}) and (\ref{D-1-}) to $TT^*T^{k-1}Q$. This reads  
\begin{equation}\label{LagSub2ndLag}
({q}_{(\kappa)}^A,{p^{(\kappa)}_A};{q}_{(1)}^A,\dots, q_{(k)}^A, \frac{\partial L}{\partial {q}_{(0)}^A}, \frac{\partial L}{\partial q_{(1)}^A}-p^{(0)}_A, \dots, 
\frac{\partial L}{\partial q_{(k-1)}^A}-p^{(k-2)}_A ). 
\end{equation}
The submanifold in (\ref{D-1--}) and (\ref{D-1-})
is a Lagrangian submanifold of $TT^*T^{k-1}Q$. The dynamics in this submanifold is represented by a systems of implicit differential equations
\begin{eqnarray} \label{impl-Ost}
\begin{split}
\dot{q}_{(0)}^A&={q}_{(1)}^A,\dots, \dot{q}_{(k-1)}^A=q_{(k)}^A, \quad \dot{p}^{(0)}_A=\frac{\partial L}{\partial {q}_{(0)}^A},\\ \dot{p}^{(1)}_A&=\frac{\partial L}{\partial q_{(1)}^A}-p^{(0)}_A, \dots, \dot{p}^{(k-1)}_A=\frac{\partial L}{\partial q_{(k-1)}^A}-p^{(k-2)}_A ,\quad 
p^{(k-1)}_A- \frac{\partial L}{\partial  q_{(k)}^A}={0}
\end{split}
\end{eqnarray}
equipped with the constraints given in \eqref{D-1-}.
It is immediate now to check that the Lagrangian submanifold (\ref{D-1--}) and (\ref{D-1-}), or the system of implicit equations (\ref{impl-Ost}) correspond to the higher order Euler-Lagrange equations
\begin{equation}\label{ele}
\frac{\partial L}{\partial
q_{(0)}^A}-\frac{d}{dt}\frac{\partial L}{\partial {q}_{(1)}^A}+
\frac{d^{2}}{dt^{2}}\frac{\partial L}{\partial {q}_{(2)}^A}\dots+(-1)^k\frac{d^k}{dt^k}\frac{\partial L}{\partial {q}_{(k)}^A} =0.
\end{equation}
Note that these identifications are independent of the regularity of the Lagrangian functions. 
\bigskip 

\noindent  \textbf{{Hamilton-Jacobi Equations.}}
Introduce a closed one-form $\gamma$ on $T^{k-1}Q$ given locally by 
\begin{equation} \label{gamma-k-1}
\gamma= \gamma^{(\kappa)}_A dq^A_{(\kappa)} = {\gamma}^{(0)}_A dq^A_{(0)}+ {\gamma}^{(1)}_A dq^A_{(1)}+\dots + {\gamma}^{(k-1)}_A dq^A_{(k-1)}.
\end{equation}
Now we apply the implicit Hamilton-Jacobi Theorem (\ref{nHJT}) to the first order implicit system given in \eqref{impl-Ost}, which is equivalent to the higher order Euler-Lagrange system. More concretely, we are employing the second condition (\ref{IHHJE}) in Theorem (\ref{nHJT}) to the present case. This reads 
$$d\left({\gamma^{(0)}_A} {q_{(1)}^A}+{\gamma^{(1)}_A} {q_{(2)}^A}
+{\gamma^{(k-1)}_A} {q_{(k)}^A}
-L({q}_{(0)},{q_{(1)}},\dots,{q_{(k)}})\right)=0.$$
Accordingly, we compute the following system of equations
\begin{equation}\label{HJfor2ndOr2--}
\begin{split}
\frac{\partial \gamma^{(0)}_A} {\partial q^B_{(0)} } {q_{(1)}^A}
+\dots +
\frac{\partial \gamma^{(k-1)}_A} {\partial q^B_{(0)} } q_{(k)}^A - 
\frac{\partial L}{\partial q^B_{(0)}} =0,
\\
\frac{\partial \gamma^{(0)}_A} {\partial q^B_{(1)} } {q_{(1)}^A}
+ \gamma^{(0)}_B+ \dots +
\frac{\partial \gamma^{(k-1)}_A} {\partial q^B_{(1)} } q_{(k)}^A - 
\frac{\partial L}{\partial q^B_{(1)}} =0,
 \\
 \dots
 \\
 \frac{\partial \gamma^{(0)}_A} {\partial q^B_{(k-1)} } {q_{(1)}^A}
+ \frac{\partial \gamma^{(1)}_A} {\partial q^B_{(k-1)} } {q_{(1)}^A} + \dots + \gamma^{(k-2)}_{B}+
\frac{\partial \gamma^{(k-1)}_A} {\partial q^B_{(k-1)} } q_{(k)}^A - 
\frac{\partial L}{\partial q^B_{(1)}} =0
 \\
\gamma^{(k-1)}_A-\frac{\partial L}{\partial q^A_{(k)}} =0.
\end{split}
\end{equation}

Since $\gamma$ is a closed one-form, then in a local chart, one may take $\gamma$ as the exterior derivative $dW$ of a real-valued function $W$ on $T^{k-1}Q$. In this case, we integrate the system as 
\begin{equation}
\frac{\partial W}{\partial q^A_{(0)}} {q_{(1)}^A}+\frac{\partial W}{\partial q^A_{(1)}}  {q_{(2)}^A}+\dots+\frac{\partial W}{\partial q^A_{(k-1)}} {q_{(k)}^A}-L({q}_{(0)},{q_{(1)}},\dots, {q_{(k)}})=0, \qquad \frac{\partial W}{\partial q^A_{(k-1)}} - \frac{\partial L}{\partial q^A_{(k)}} =0.
\end{equation}

\bigskip
\noindent  \textbf{{Hamilton-Jacobi equations for nondegenerate cases.}}
Note that, we may solve the Lagrange multipliers ${q_{(k)}^A}$ from the definition of conjugate momenta $p^{(k-1)}_A$ using the constraint (\ref{D-1-}) if the matrix $\left[{\partial^2 L}/{\partial q_{(k)}^A \partial q_{(k)}^B}\right]$ is nondegenerate. In this case, the solution has the form $$q_{(k)}^A= \Sigma^A\left({q}_{(0)},q_{(1)},...,q_{(k-1)},{p}^{(1)}\right).$$
Further, in a local chart, one may take $\gamma$ as the exterior derivative $dW$ of a real-valued function $W$ on $TQ$. In this case the requirement that $E$ is constant on the image of $dW$ results in a Hamilton-Jacobi equation in form
\begin{equation}\label{HJ-SO}
\frac{\partial W}{\partial q^A_{(0)}} {q_{(1)}^A}+\frac{\partial W}{\partial q^A_{(1)}}  {q_{(2)}^A} +\dots+\frac{\partial W}{\partial q^A_{(k-1)}} {\Sigma}^A -L({q}_{(0)},{q_{(1)}},,\dots,q_{(k-1)},{\Sigma})=0.
\end{equation}

\subsubsection{The Morse family method - the Schmidt-Legendre transformation for second order systems} \label{S-2-HJ}

We start by recalling some basics on the acceleration bundle and refer the reader to \cite{EsGu18} for further details. 
\bigskip

\noindent \textbf{{Acceleration bundle.}} Consider the set $K_q( Q)$ of smooth curves passing through $q\in Q$ whose first derivatives vanish at $q$, that is
\begin{equation}
K_q(Q)=\left\{\gamma\in C_{q}(Q):D(f\circ \gamma) (0)=0, \quad \forall f:Q\mapsto\mathbb{R} \right\}.
\end{equation}
Define an equivalence relation on $K_q(Q)$ by saying that two curves $\gamma$ and $\gamma'$ are equivalent if the second derivatives of $\gamma$ and $\gamma'$ are equal at the point $q$, that is if
\begin{equation*}
\gamma(0)=\gamma'(0)=q, \qquad D^2(f\circ \gamma) (0)=D^2 (f\circ\gamma') (0), \qquad \forall f:Q\mapsto\mathbb{R}
\end{equation*}
for all real valued functions $f$ on $Q$.  An equivalence class is denoted by $\mathfrak{a}\gamma(0)$. The set of all of these equivalence classes is called  acceleration space $A_{q}Q$ at $q\in Q$.
If $Q$ is an $n$-dimensional manifold then union of all acceleration spaces
\begin{equation*}
AQ=\bigsqcup_{q\in Q}
A_{q}Q.
\end{equation*}
is a $2n$-dimensional manifold called as the acceleration bundle of $Q$. The induced local coordinates on $AQ$ are defined to be 
\begin{align}
( {q}_{(0)}^A,     {a}_{(0)}^A     ):AQ\longrightarrow \mathbb{R}^{2n}:&\mathfrak {a}\gamma(0)\longrightarrow (  {q}_{(0)}^A \circ \gamma(0),D^2(  {q}_{(0)}^A \circ \gamma)(0)).
\end{align}
We note that, the third order tangent bundle $T^{3}Q$ and the tangent bundle of the acceleration
bundle are isomorphic. If the induced coordinates assumed on the tangent bundle $TAQ$ are  $\left( {q}_{(0)}^A, {a}_{(0)}^A;{q}_{(1)}^A, {a}_{(1)}^A\right)$, then the isomorphism $S$ locally takes the form
\begin{equation}
S:TAQ\rightarrow T^{3}Q:\left( {q}_{(0)}^A, {a}_{(0)}^A;{q}_{(1)}^A, {a}_{(1)}^A\right)
\rightarrow\left( {q}_{(0)}^A, {q}_{(1)}^A; {a}_{(0)}^A, {a}_{(1)}^A\right).
\label{iso}
\end{equation}

\bigskip
\noindent  \textbf{{Gauge invariance of the Lagrangian formalism and Schmidt's method.}}
Consider a second order Lagrangian function%
\begin{equation}
{L=L}\left( q_{(0)},q_{(1)},q_{(2)}\right)  \label{Lag}
\end{equation}
on $T^{2}Q$. The gauge invariance of the second order Euler-Lagrange equations implies that the equations of motion generated by $L$ and $L+(d/dt)F$ are the same for any smooth function $F$ on $T^2Q$. When we consider $F$, we come up with a third order Lagrangian  
\begin{eqnarray}
\begin{split}
\hat{L}\left(  q_{(0)},q_{(1)},q_{(2)}, {q}_{(3)}\right) ={L}%
\left( q_{(0)},q_{(1)},q_{(2)} \right) +\frac{d}{dt}F\left(
q_{(0)},q_{(1)},q_{(2)} \right) \\ ={L}\left( q_{(0)},q_{(1)},q_{(2)} \right) +\frac{\partial F}{%
\partial q_{(0)}^A }{q_{(1)}^A }+\frac{\partial F}{\partial q_{(1)}^A }{q_{(2)}^A}+\frac{\partial F}{\partial q_{(2)}^A }{q_{(3)}^A}.  \label{Lhat}
\end{split}
\end{eqnarray}
defined in $T^3Q$ with local coordinates $\left( q_{(0)},q_{(1)},q_{(2)},q_{(3)}\right)$.
By recalling the isomorphism in (\ref{iso}), we pull back the
Lagrangian $\hat{L}$ to the tangent bundle $TAQ$, so it results in a first order Lagrangian function
\begin{equation}
L_{2}:TAQ \mapsto \mathbb{R}: (  {q}_{(0)}^A, {a}_{(0)}^A;{q}_{(1)}^A, {a}_{(1)}^A) \mapsto  {L}\left( q_{(0)},q_{(1)},a_{(0)} \right) +\frac{\partial F}{%
\partial q_{(0)}^A }{q_{(1)}^A }+\frac{\partial F}{\partial q_{(1)}^A }{a_{(0)}^A}+\frac{\partial F}{\partial a_{(0)}^A }{a_{(1)}^A}  \label{L2}
\end{equation}
defined on the first order tangent bundle $TAQ$. The Euler-Lagrange equations generated by $L_{2}$ are computed to be
\begin{equation}
\frac{\partial L_{2}}{\partial  q_{(0)}^A  }-\frac{d}{dt}\frac{\partial L_{2}}{%
\partial q_{(1)}^A}={0}, \qquad \frac{\partial L_{2}}{\partial  a_{(0)}^A  }-\frac{d}{dt}\frac{\partial L_{2}}{%
\partial a_{(1)}^A}={0}.  \label{ELtwo}
\end{equation}
The second set of equations in (\ref{ELtwo}) can be rewritten as
\begin{equation}
\left( \frac{\partial L}{\partial a_{(0)}^A }+\frac{\partial F}{
\partial q_{(1)}^A }\right) +\frac{\partial^{2}F}{\partial q_{(1)}^A \partial{     a_{(0)}^B}}(     a_{(0)}^B-
q_{(2)}^B)={0}.  \label{cons}
\end{equation}
Assume that the second order Lagrangian function $L$ is a nondegenerate, that is the rank of the Hessian matrix $\left [{\partial^{2}L}/{
\partial {a}_{(0)} ^A \partial {a}_{(0)} ^B } \right ]$ is maximal, and assume also that the
auxiliary function $F$ satisfies
\begin{equation}
\frac{\partial L}{\partial a_{(0)}^A }+\frac{\partial F}{
\partial q_{(1)}^A} =0.  \label{chi}
\end{equation}
In this case, the non-degeneracy of the matrix $\left [{\partial^{2}L}/{
\partial {a}_{(0)} ^A \partial {a}_{(0)} ^B } \right ]$ implies the non-degeneracy of the matrix $%
[\partial^{2}F/\partial{a_{(0)}^A }\partial {q_{(1)}^B}]$. Given this, the equations (\ref{cons}) reduce to the set of constraints $a_{(0)}^B-
q_{(2)}^B=0$. In this case, the first set in (\ref{ELtwo}) results in the same Euler-Lagrange equations generated by $L$ in (\ref{Lag}). 

\bigskip
\noindent  \textbf{{Morse family generating the Lagrangian submanifold.}}
Assuming the dual coordinates $\left( q^A_{(0)},{a}^A_{(0)},p_A^{(0)},{\pi}_A^{(0)}\right)$ on the cotangent bundle $T^*AQ$, define the following Morse family
\begin{align}  \label{HMF}
E&\left( q_{(0)},{a}_{(0)},p^{(0)},{\pi}^{(0)},{q_{(1)}},{ a_{(1)}}\right)=
p_A^{(0)}  {q_{(1)}^A}+{\pi}_A^{(0)}{ a_{(1)}^A}-L_2\left( {q}_{(0)}, {a}_{(0)};{q}_{(1)}, {a}_{(1)} \right)\\
&=
p_A^{(0)}  {q_{(1)}^A}+{\pi}_A^{(0)}{ a_{(1)}^A}-{L}({q}_{(0)},{q}_{(1)}, {a}_{(0)}) -\frac{\partial F}{%
\partial q_{(0)}^A }{q_{(1)}^A }-\frac{\partial F}{\partial q_{(1)}^A }{a_{(0)}^A}-\frac{\partial F}{\partial a_{(0)}^A }{a_{(1)}^A}  \notag
\end{align}
on the Whitney sum $TAQ\times_{AQ} T^{\ast}AQ$ over the base manifold $T^{\ast}AQ$. The conjugate momenta are defined by the equations
\begin{align}\label{MorseSc12}
{0}&=\frac{\partial E}{\partial {q_{(1)}^A}}=%
p_A^{(0)}  -\frac{\partial L_2}{\partial {q_{(1)}^A}},
\qquad {0%
}=\frac{\partial E}{\partial { a_{(1)}^A}}={\pi}_A^{(0)}-%
\frac{\partial L_2}{\partial { a_{(1)}^A}}={\pi}_A^{(0)}-\frac{\partial F}{\partial a_{(0)}^A}.
\end{align}
If we substitute the momenta ${\pi}_A^{(0)}$ in the definition of the Morse family (\ref%
{HMF}) which makes the family free of $(a_{(1)}^A)$, it results in
\begin{equation} \label{MorseSc1}
E(  q_{(0)},{a}_{(0)},p^{(0)},{\pi}^{(0)},{q_{(1)}})=
p_A^{(0)}  {q_{(1)}^A}-{L}({q}_{(0)},{q}_{(1)}, {a}_{(0)}) -\frac{\partial F}{%
\partial q_{(0)}^A }{q_{(1)}^A }-\frac{\partial F}{\partial q_{(1)}^A }{a_{(0)}^A}
\end{equation}
defined on the Whitney sum $T^*AQ\times_Q TAQ$.
A further reduction on the Morse family is possible. For this, recall the assumption that the matrix $%
[\partial^{2}F/\partial{a_{(0)}^A }\partial {q_{(1)}^B}]$ is nondegenerate. So that we can, at least locally, solve $q_{(1)}^A$ in terms of the momenta from the second equation ${\pi}_A^{(0)}={\partial F} (
q_{(0)},q_{(1)},q_{(2)} )/{\partial a_{(0)}^A}$ in (\ref{MorseSc12}). Let us write this solution as
\begin{equation}
{q_{(1)}^A}={z}^A\left(   q_{(0)}  ,a_{(0)},{\pi}^{(0)}\right) .
\label{Z}
\end{equation}
This results with a well-defined Hamiltonian function
\begin{equation}\label{canH}
H\left(  q_{(0)},{a}_{(0)};p^{(0)},{\pi}^{(0)}\right) =p^{(0)}_A{z^A}
-L\left( { q_{(0)} ,z, a_{(0)} }\right) -\frac{\partial F}{%
\partial q_{(0)}^A }{z^A }-\frac{\partial F}{\partial q_{(1)}^A }{a_{(0)}^A}    
\end{equation}
 on $T^*AQ$.

\bigskip
\noindent  \textbf{{Hamilton-Jacobi theory in the acceleration bundle framework.}}
Now, we are ready to write the Hamilton-Jacobi theory for second order nondegenerate Lagrangian functions. For this, assume a real valued function $W$ defined on the acceleration bundle $AQ$ and a Hamiltonian vector field $X_H$ on $T^{*}AQ$ associated to the Hamiltonian function $H$ in (\ref{canH}). We can define a vector field $X_H^{\gamma}$ on the acceleration bundle $AQ$
\begin{equation}\label{gammarelated}
 X_H^{\gamma}=T\pi_{AQ}\circ X_H\circ \gamma.
\end{equation}
according to the commutativity of the diagram
\begin{equation}\label{Xg}
  \xymatrix{ T^{*}AQ
\ar[dd]^{\pi_{AQ}} \ar[rrr]^{X_H}&   & &TT^{*}AQ\ar[dd]^{T\pi_{AQ}}\\
  &  & &\\
 AQ\ar@/^2pc/[uu]^{\gamma=dW}\ar[rrr]^{X_H^{\gamma}}&  & & TAQ }
\end{equation}
\begin{theorem}[HJ theorem in the acceleration bundle]\label{HJacc}
Let $\gamma=dW$ be a closed one-form on $AQ$, we say that $\gamma$ is a solution of the Hamilton--Jacobi problem in the acceleration bundle if the following two equivalent conditions are satisfied
\begin{enumerate}
\item The vector fields $X_{H}$ and $X_{H}^{\gamma }$ are $\gamma$-related
\item $d\left( H\circ \gamma \right)=0.$
\end{enumerate}
\end{theorem}
We can rewrite the second condition as an equation the function $W=W({q_{(0)},a_{(0)}})$ satisfying the partial differential equation
\begin{equation} \label{HJ-Sch}
\frac{\partial W}{\partial q_{(0)}^A}{z^A}\left(   q_{(0)}  ,a_{(0)},\frac{\partial W}{\partial q_{(0)}}\right)
-L\left( {q_{(0)},z,a_{(0)}}\right) -\frac{\partial F}{\partial  q_{(0)}^A }{z^A}\left(   q_{(0)}  ,a_{(0)},\frac{\partial W}{\partial q_{(0)}}\right)-\frac{\partial F}{\partial q_{(1)}^A }{a_{(0)}^A}    =E,      
\end{equation}
where $E$ is a constant.

Let us write the second condition explicitly for a particular case. 
Determine the auxiliary function $F(     {a} _{(0)}    ,{q_{(1)}})=    - \delta_{AB} {a}_{0} ^A     {q_{(1)}^B}$ in (\ref{L2}). In the light of condition (\ref{chi}), the Lagrangian is quadratic with respect to  second order time derivatives. More concretely, we see that $\partial L / \partial {a^A_{(0)}} =\delta_{AB} a^B_{(0)}$. In this case, the Hamiltonian function (\ref{canH}) reduces to
 \begin{equation}
 H\left(  q_{(0)},{a}_{(0)};p^{(0)},{\pi}^{(0)}\right) = \delta^{AB}p^{(0)}_A \pi^{(0)}_B
-L\left( { q_{(0)} ,\pi^{(0)}, a_{(0)} }\right) 
+\delta_{AB} {a}_{(0)} ^A {a}_{(0)} ^B .
\end{equation}
In this case, for a closed one-form $$\gamma=\frac{\partial W}{ \partial  q_{(0)}^A} dq_{(0)}^A+ \frac{\partial W}{ \partial a_{(0)}^A } da_{(0)}^A,$$ the second condition in Theorem \eqref{HJacc} provides the following Hamilton-Jacobi equation for the nondegenerate second order Lagrangian function $L$ 
\begin{equation}
\delta^{AB}\frac{\partial W}{ \partial  q_{(0)}^A} \frac{\partial W}{ \partial  q_{(0)}^B} -L\left( {q}_{(0)},\frac{\partial W}{\partial a_{(0)}},{a}_{(0)}\right) +\delta_{AB} {a}_{(0)} ^A {a}_{(0)} ^B=E.
\end{equation} 

\subsubsection{Comparisons of HJ formalisms for nondegenerate cases} 

Let us consider again  the auxiliary function $F=F({q_{(0)},q_{(1)},q_{(2)}})$ on the second order tangent bundle $T^2Q$ and let us write $T^2Q$ locally as a product space $AQ\times_{Q} TQ$. Here, the function will have a form $F=F({q_{(0)},q_{(1)},a_{(0)}})$.  In this case, the cotangent bundle of $T^*T^2Q$ can be identified with the product space $T^*AQ\times T^*TQ$. The image of the exterior derivative $dF$ determines a Lagrangian submanifold of $T^*AQ\times T^*TQ$ hence a symplectic diffemorphism between $T^*AQ$ and $T^*TQ$. Explicitly, the symplectic
diffeomorphism is computed to be%
\begin{eqnarray} \label{sym-diff}
T^{\ast} AQ\rightarrow T^{\ast}TQ: \left(   q_{(0)}^A,{a}_{(0)}^A;p^{(0)}_A,{\pi}^{(0)}_A\right) \rightarrow  \\ \left(    q_{(0)}^A  ,%
{z}^A\left(   q_{(0)}  ,{a}_{(0)},{\pi}^{(0)}\right) ,p^{(0)}_A-\frac{\partial F}{\partial   q_{(0)}^A }%
\left(   q_{(0)}  ,{z}\left(   q_{(0)}  , {a}_{(0)},
{\pi}^{(0)}\right) ,{a}_{(0)}\right) ,-\frac{\partial F}{%
\partial{z}^A}\right).\notag
\end{eqnarray}
This symplectic diffeomorphism establishes the link between the Morse families  \eqref{Morse-Ost} (when $k=1$) and \eqref{MorseSc1}. To see this directly, let us now pull back the Morse family $E$ given in (\ref{Morse-Ost}) by the mapping (\ref{sym-diff}). We compute the result as follows
\begin{eqnarray}
{p}^{(0)} {q_{(1)}}+{p}^{(1)} {q_{(2)}}-L({q}_{(0)},{q_{(1)}},{q_{(2)}})\nonumber\\ = \left({p}_{q}-\frac{\partial F}{\partial  {q}  }%
\left(   {q}  ,{z}\left(   {q}  , {a},
{p}_{a}\right),{a}\right)\right) {z}-\frac{\partial F}{%
\partial{z}} {a}-L({q},{z},{a})\nonumber\\
= {p}_{q} {z}-L({q},{z},{a})-
\frac{\partial F}{\partial  {q}  }%
{z}-\frac{\partial F}{\partial{z}}     {a} 
\end{eqnarray}
which is exactly the Hamiltonian function in (\ref{canH}). Here, we have employed the identification $a_{(0)}^A=q_{(2)}^A$. The following examples compare the two methods we have exhibited so far.

\begin{example} \label{example}\normalfont
Let us consider a pure quadratic one-dimensional Lagrangian 
\begin{equation}\label{example5}
L=\frac{1}{2}\mu q_{(2)}^2
\end{equation}
 If we first apply the Ostragradski method, the momentum ${p}^{(1)}$ is computed to be $\mu{q_{(2)}}$. The Hamilton-Jacobi equation (\ref{HJ-SO}) for this system is 
\begin{equation}
\frac{\partial W}{ \partial q_{(0)}} {q_{(1)}}+ \frac{1}{2}\mu\left(\frac{\partial W}{\partial q_{(0)}}\right)^2=c.
\end{equation}

Let us now apply the Schmidt's method presented in Section (\ref{S-2-HJ}) to the Lagrangian \eqref{example5}. Condition (\ref{chi}) integrates the function $F$ as
\begin{equation}
 F=-\mu {a} {q_{(1)}}+g(q_{(0)})
 \end{equation}
  where $g$ is an arbitrary function which can be chosen as zero without loss of any generality. This enables us to use the Hamilton-Jacobi equation in (\ref{HJ-Sch}), which is exactly
\begin{equation}
-\frac{\partial W}{\partial q_{(0)}}\frac{\partial W}{\partial a_{(0)}}+\frac{1}{2}\mu a_{(0)}^2=c
\end{equation}
and which can be solved assuming that $\nabla_{a}W$ does not equal to zero, and rewrite the Hamilton-Jacobi problem in the form
 \begin{equation}
\frac{\partial W}{ \partial q_{(0)}} =\frac{\frac{1}{2}\mu a_0^2-c}{\frac{\partial W}{ \partial a_{(0)}}}=c_2
 \end{equation}
where $c_2$ is a constant. Its solution reads:
\begin{equation}
W(q_{(0)},a_{(0)})=c_2q_{(0)}+\frac{1}{6c_2}\mu a_{(0)}^3-\frac{c}{c_2}a_{(0)}.
\end{equation}
\end{example}

\subsubsection{The Morse family method - The Schmidt's method for the third order Lagrangians}
\label{tol} 

Let us start with a third order Lagrangian function ${L}( q_{(0)},q_{(1)},q_{(2)}, {q}_{(3)})$ defined on $T^3Q$. Recalling the local diffeomorphism in (\ref{iso}), we pull back the Lagrangian function $L$ to the tangent bundle $TAQ$ of the acceleration bundle. By this, we arrive at a first order Lagrangian function $L=L\left(q_{(0)},a_{(0)};q_{(1)},a_{(1)}\right) $. Now, we define a manifold $M$ with local coordinates ${m}$,  its tangent bundle $TM$ with coordinates $({m_{(0)}^A, m_{(1)}^B})$ and the first order Lagrangian function
\begin{equation}
L_{3} =L\left(q_{(0)},a_{(0)};q_{(1)}, a_{(1)}\right)+\frac{\partial F}{%
\partial q_{(0)}^A }{q_{(1)}^A }+\frac{\partial F}{\partial q_{(1)}^A }{a_{(0)}^A}+\frac{\partial F}{\partial a_{(0)}^A }{a_{(1)}^A} +\frac{\partial F}{%
\partial m_{(0)}^A }{m_{(1)}^A }.  \label{L3}
\end{equation}
on the tangent bundle $T(AQ\times M)$ equipped with local coordinates $$( {q}_{(0)}^A, {a}_{(0)}^A;{q}_{(1)}^A, {a}_{(1)}^A;{m}_{(0)}^A,{m}_{(1)}^A).$$ Here, the auxiliary function $F$ depends on $(
{q}_{(0)},{q}_{(1)},{a}_{(0)},{m}_{(0)})$. The Euler Lagrange equations generated by the Lagrangian $L_{3}$ are equal to the Euler-Lagrange equations generated by the third order Lagrangian function $L$ if the requirement
\begin{equation}
\det[\partial^{2}F/\partial{{q}_{(1)}^A}\partial {m}_{(0)}^B]\neq  0
\label{cond2}
\end{equation}
is assumed \cite{EsGu18}. 

We consider the conjugate momenta on $T^*(AQ\times M)$ determined locally by $(p^{(0)}_A,\pi^{(0)}_A,\mu^{(0)}_A)$ and the energy function associated with $L_3$ is 
\begin{eqnarray}
E&=& p^{(0)}_A {q}_{(1)}^A+\pi^{(0)}_A{a}_{(1)}^A+\mu^{(0)}_A{m}_{(1)}^A-L_3  \label{totHam--}
\\&=& p^{(0)}_A {q}_{(1)}^A+\pi^{(0)}_A{a}_{(1)}^A+\mu^{(0)}_A {m}_{(1)}^A-L-\frac{\partial F}{%
\partial q_{(0)}^A }{q_{(1)}^A }-\frac{\partial F}{\partial q_{(1)}^A }{a_{(0)}^A}  \notag \\&& \qquad \qquad -\frac{\partial F}{\partial a_{(0)}^A }{a_{(1)}^A} -
\frac{\partial F}{\partial m_{(0)}^A }{m_{(1)}^A}. \notag
\end{eqnarray}
Notice that this energy function is a Morse family on the Whitney sum $T(AQ\times M) \times T^*(AQ\times M)$ and, in accordance with the following diagram.
 \begin{equation}\label{Ham-Morse-Gen---}
\xymatrix{TT^{\ast }(AQ\times M)  \ar[dr]_{\tau_{T^{\ast }(AQ\times M)}\quad}
&& \ar[ll]_{\Omega_{(AQ\times M)}^{\sharp}} T^{\ast }T^{\ast }(AQ\times M)\ar[dl]^{\quad \pi _{T^{\ast}(AQ\times M)}} &T(AQ\times M) \times T^*(AQ\times M) \ar[d]^{\pi_2}
\\&T^{\ast}(AQ\times M)  \ar@{=}[rr]&& T^{\ast}(AQ\times M)},
\end{equation}
The family $E$ generates a Lagrangian submanifold ${D}$ of $T^\ast T^{\ast }(AQ\times M)$, and using the musical isomorphism $\Omega_{(AQ\times M)}^{\sharp}$, we map this Lagrangian submanifold to a Lagrangian submanifold $S$ of $TT^{\ast }(AQ\times M)$, that is a symplectic manifold equipped with the lifted symplectic two-form $\Omega_{(AQ\times M)}^T$. This Lagrangian submanifold exactly determines the third order Euler-Lagrange equations generated by the Lagrangian $L=L({q_{(0)},q_{(1)},q_{(2)}, q_{(3)}})$. 

Let us now apply the implicit Hamilton-Jacobi theorem to this case. Assume a closed one-form $\gamma$ on $AQ\times M$ given locally by
\begin{equation}
\gamma=\gamma_A dq^A_{(0)}+\alpha_A da^A_{(0)}+\beta_A d{m}^A_{(0)}.
\end{equation}
The restriction of the Lagrangian submanifold $S$ to the image space of $\gamma$ will be denoted by $S|_{Im(\gamma)}$. Then project $S|_{Im(\gamma)}$ to the tangent bundle $T(AQ\times M)$ by means of the tangent mapping $T\pi_{AQ\times M}$. This results in a (possibly non horizontal) submanifold $S^\gamma=T_{\pi_{AQ\times M}}(S|_{Im})$ of $T(AQ\times M)$. Let us depict these in the following diagram
\begin{equation}\label{Diag-HJ-S-3}
  \xymatrix{ T^{*}(AQ\times M)
\ar[dd]^{\pi_{AQ\times M}} \ar[rrr]&   & &TT^{*}(AQ\times M)\ar[dd]^{T\pi_{AQ\times M}}& S_{Im(\gamma)}\ar@{^{(}->}[l]\\
  &  & &\\
 AQ\times M\ar@/^2pc/[uu]^{\gamma=dW}\ar[rrr]&  & & T(AQ\times M) &S^\gamma\ar@{^{(}->}[l]}
\end{equation}

\begin{theorem}\label{HJT3A} (\textbf{HJ theorem for implicit third order Lagrangians in the  acceleration space}) 
A solution of the implicit Hamilton--Jacobi problem for third order Lagrangians in the acceleration space is a closed one-form $\gamma$ that fulfills the two following equivalent relations: 
\begin{enumerate}
\item The Lagrangian submanifold $S\vert_{Im({\gamma})}$ and the submanifold $S^{\gamma}$ are $\gamma$-related, that is $T\gamma(S^{\gamma})=S\vert_{Im({\gamma})}$.
\item $d(E({q}_{(0)}^A, {a}_{(0)}^A,{m}_{(0)}^A;\gamma_A,\alpha_A,\beta_A;{q}_{(1)}^A, {a}_{(1)}^A,{m}_{(1)}^A))=0$.
\end{enumerate} 

\end{theorem}

The second condition reads the implicit HJ equation
$$E=\gamma_A {q}_{(1)}^A+\alpha_A {a}_{(1)}^A+\beta_A {m}_{(1)}^A-L_3=c,$$
where $c$ being a constant. Taking the exterior derivative of this equation, we arrive at the following local picture of the Hamilton-Jacobi equation  \begin{eqnarray}\label{HJ-L-3}
 \begin{split}
 \frac{\partial \gamma_A}{\partial{q}_{(0)}^B} {q}_{(1)}^A 
+\frac{\partial \alpha_A}{\partial{q}_{(0)}^B} {a}_{(1)}^A
+ \frac{\partial \beta_A}{\partial{q}_{(0)}^B} {m}_{(1)}^A-\frac{\partial L_3}{\partial{q}_{(0)}^B}=0 
\\
\frac{\partial \gamma_A}{\partial{a}_{(0)}^B} {q}_{(1)}^A 
+\frac{\partial\alpha_A }{\partial{a}_{(0)}^B} {a}_{(1)}^A
+ \frac{\partial \beta_A}{\partial{a}_{(0)}^B} {m}_{(1)}^A-\frac{\partial L_3}{\partial{a}_{(0)}^B}=0 
 \\
\frac{\partial \gamma_A}{\partial{m}_{(0)}^B} {q}_{(1)}^A 
+\frac{\partial\alpha_A}{\partial{m}_{(0)}^B} {a}_{(1)}^A
+ \frac{\partial \beta_A}{\partial{m}_{(0)}^B} {m}_{(1)}^A-\frac{\partial L_3}{\partial{m}_{(0)}^B}=0
 \\
\gamma_A-\frac{\partial L_3}{\partial{q}_{(1)}^A}=0, 
 \\
\alpha_A-\frac{\partial L_3}{\partial{a}_{(1)}^A}=0, 
  \\
\beta_A-\frac{\partial L_3}{\partial{m}_{(1)}^A}=0, 
 \end{split}
\end{eqnarray}  
 where $L_3$ is the Lagrangian function in (\ref{L3}). As a particular case, we consider that the auxiliary function is taken to be $F=\delta_{AB} {q}_{(1)}^A {m}_{(0)}^B$. In this case the Lagrangian function $L_3$ reduces to 
 \begin{equation}
L_{3}\left( {q}_{(0)}, {a}_{(0)};{q}_{(1)}, {a}_{(1)};{m}_{(0)},{m}_{(1)}\right ) =L\left({q}_{(0)}, {a}_{(0)},{q}_{(1)},{a}_{(1)} \right)
+\delta_{AB} {q}_{(1)}^A {m}_{(0)}^B. \label{red-L3}
\end{equation}
In this case, the last equation in system \eqref{HJ-L-3} provides the definition of the Lagrange multiplier as $q_{(1)}^{A}=\delta^{AC}\beta_C$. So that the substitution of the Lagrangian \eqref{red-L3} into \eqref{HJ-L-3}, we the following reduced Hamilton-Jacobi equations
 \begin{eqnarray} \label{HJ-L-3-red}
  \begin{split}
  \delta^{AC}\gamma_C\frac{\partial \gamma_A}{\partial{q}_{(0)}^B} 
+\frac{\partial\alpha_A}{\partial{q}_{(0)}^B} {a}_{(1)}^A
+ \frac{\partial \beta_A}{\partial{q}_{(0)}^B} {m}_{(1)}^A-\frac{\partial L}{\partial{q}_{(0)}^B}\bigg\rvert_{{q_{(1)}^A}=\delta^{AB}\beta_B}=0 
\\
\delta^{AC}\gamma_C \frac{\partial \gamma_A}{\partial{a}_{(0)}^B} 
+\frac{\partial\alpha_A}{\partial{a}_{(0)}^B} {a}_{(1)}^A
+ \frac{\partial \beta_A}{\partial{a}_{(0)}^B} {m}_{(1)}^A-\frac{\partial L}{\partial{a}_{(0)}^B}\bigg\rvert_{{q_{(1)}^A}=\delta^{AB}\beta_B}=0 
 \\
\delta^{AC}\gamma_C\frac{\partial \gamma_A}{\partial{m}_{(0)}^B} 
+\frac{\partial\alpha_A}{\partial{m}_{(0)}^B} {a}_{(1)}^A
+ \frac{\partial \beta_A}{\partial{m}_{(0)}^B} {m}_{(1)}^A-\frac{\partial L}{\partial{m}_{(0)}^B}\bigg\rvert_{{q_{(1)}^A}=\delta^{AB}\beta_B}=0
 \\
\gamma_A-\frac{\partial L}{\partial{q}_{(1)}^A}\bigg\rvert_{{q_{(1)}^A}=\delta^{AB}\beta_B}-\delta_{AB}m_{(1)}^B=0,
 \\
\alpha_A-\frac{\partial L}{\partial{a}_{(1)}^A}\bigg\rvert_{{q_{(1)}^A}=\delta^{AB}\beta_B}=0, 
 \end{split}
 \end{eqnarray}

\bigskip

\noindent  \textbf{{Hamilton-Jacobi Theory for degenerate second order Lagrangians.}}
Notice that up to now, the non-degeneracy condition has not been assumed. This implies that we can apply this framework in both degenerate and nondegenerate third order Lagrangian systems. It is also interesting to note that we can further study the second order Lagrangian systems in the present framework. Let us study this particular case. In the definition of $L_3$ given in (\ref{L3}), we choose $L=L(   {q}_{(0)}, {a}_{(0)},{q}_{(1)})$, and consider an auxiliary function $F=F({ {q}_{(0)},q_{(1)}, {m}_{(0)}}) $. So that, we have a Lagrangian function
\begin{equation}
L_{2-deg}\left(   {q}_{(0)}, {a}_{(0)};{q}_{(1)}, {a}_{(1)};{m}_{(0)},{m}_{(1)}\right) = L \left(   {q} _{(0)} ,{q_{(1)}}, {a}_{(0)}   \right) +
\frac{\partial F}{%
\partial q_{(0)}^A }{q_{(1)}^A }+\frac{\partial F}{\partial q_{(1)}^A }{a_{(0)}^A}
+\frac{\partial F}{%
\partial m_{(0)}^A }{m_{(1)}^A }
\label{L-2-deg}
\end{equation}
defined on the tangent bundle $T(AQ\times M)$. In this case, the energy function (\ref{totHam--}) is reduced to 
\begin{equation}
E= p^{(0)}_A {q}_{(1)}^A+\pi^{(0)}_A{a}_{(1)}^A+\mu^{(0)}_A {m}_{(1)}^A-L-\frac{\partial F}{%
\partial q_{(0)}^A }{q_{(1)}^A }-\frac{\partial F}{\partial q_{(1)}^A }{a_{(0)}^A}-
\frac{\partial F}{\partial m_{(0)}^A }{m_{(1)}^A}. 
\label{totHam----}
\end{equation}
This Morse family generates a nonhorizontal Lagrangian submanifold of $TT^*(AQ\times M)$. So that defines an implicit Hamiltonian system. We substitute the Lagrangian $L_{2-deg}$ into the Hamilton-Jacobi equation (\ref{HJ-L-3}). This gives the following Hamilton-Jacobi equation for a second order degenerate Lagrangian $L$. The fifth equation gives us that $\alpha_A=0$. Under the light of the closure of the differential form $\gamma$, this reads 
$\gamma_{A}=\gamma_{A}({q_{(0)},m_{(0)}})$ and that $\beta_{A}=\beta_{A}({q_{(0)},m_{(0)}})$ so we have
\begin{equation} \label{newHJ-L-2-deg-22}
\begin{split}
&\frac{\partial \gamma_A}{\partial{q}_{(0)}^B} {q}_{(1)}^A 
+ \frac{\partial \beta_A}{\partial{q}_{(0)}^B} {m}_{(1)}^A=\frac{\partial L}{\partial{q}_{(0)}^B} + \frac{\partial^2 F}{\partial{q}_{(0)}^B \partial{q}_{(0)}^A} {q_{(1)}^A}
+\frac{\partial^2 F}{\partial{q}_{(0)}^B \partial{q}_{(1)}^A} {a_{(0)}^A} 
+\frac{\partial^2 F}{\partial{q}_{(0)}^B \partial{m}_{(0)}^A} {m_{(1)}^A}
,\\&\frac{\partial L}{\partial{a}_{(0)}^A}+ \frac{\partial F}{\partial{q}_{(1)}^A } =0 
 \\&
 \frac{\partial \gamma_A}{\partial{m}_{(0)}^B} {q}_{(1)}^A 
+ \frac{\partial \beta_A}{\partial{m}_{(0)}^B} {m}_{(1)}^A=
\frac{\partial^2 F}{\partial{m}_{(0)}^B \partial{q}_{(0)}^A} {q_{(1)}^A}
+ \frac{\partial^2 F}{\partial{m}_{(0)}^B \partial{q}_{(1)}^A} {a_{(0)}^A}
+\frac{\partial^2 F}{\partial{m}_{(0)}^B \partial{m}_{(0)}^A} {m_{(1)}^A} 
 \\&
 \gamma_B=\frac{\partial L}{\partial{q}_{(1)}^B}+
 \frac{\partial^2 F}{\partial{q}_{(1)}^B \partial{q}_{(0)}^A} {q_{(1)}^A}
 + \frac{\partial F}{\partial{q}_{(0)}^B}
+ \frac{\partial^2 F}{\partial{q}_{(1)}^B \partial{q}_{(1)}^A} {a_{(0)}^A}
+ \frac{\partial^2 F}{\partial{q}_{(1)}^B \partial{m}_{(0)}^A} {m_{(1)}^A} 
 \\&
  \beta_A=\frac{\partial F}{\partial{m}_{(0)}^A}. 
\end{split}
 \end{equation} 

Let us study the Hamilton-Jacobi equation (\ref{newHJ-L-2-deg-22}) for the particular choice of $F=\delta_{AB} {q}_{(1)}^A {m}_{(0)}^B$. As in the third order case, the last line of the system implies that $q_{(1)}^A=\delta^{AC}\beta_{C}$. Eventually we have  
\begin{equation} \label{HJ-L-2-deg}
\begin{split}
&\frac{\partial \delta^{AC}\beta_{C}\gamma_A}{\partial{q}_{(0)}^B}
+ \frac{\partial \beta_A}{\partial{q}_{(0)}^B} {m}_{(1)}^A=\frac{\partial L}{\partial{q}_{(0)}^B} \\
&\frac{\partial L}{\partial{a}_{(0)}^A}+ \delta_ {AB}m_{(0)}^B =0 
\\
&\delta^{AC}\beta_C \frac{\partial \gamma_A}{\partial{m}_{(0)}^B}  
+ \frac{\partial\beta_A}{\partial{m}_{(0)}^B} {m}_{(1)}^A=\delta_{BC}a_{(0)}^C 
 \\ &\gamma_B=\frac{\partial L}{\partial{q}_{(1)}^B}+\delta_{AB} {m_{(1)}^A}.
\end{split}
 \end{equation} 
 \subsection{Local vector field method}
 
The second procedure to deal with an implicit higher-order implicit Lagrangian is based on the construction of a local vector field describing the dynamics. Consider an additional section $\sigma:T^{*}T^{k-1}Q\rightarrow TT^{*}T^{k-1}Q$ in the same previous picture.

\begin{center}
 \begin{tikzcd}[column sep=tiny,row sep=huge]
  & &  &S\arrow[r, hook, "i"] &TT^*T^{k-1}Q \arrow[ld, "\tau_{T^*T^{k-1}Q}"] \arrow[rd, "T_{\pi_{T^{k-1}Q}}"] & & \\
  & &C\cap Im(\gamma)\arrow[r, hook', "i"'] &T^*T^{k-1}Q\arrow[dr,"\pi_{T^{k-1}Q}"]\arrow[ur,"\sigma",bend left=20]& &TT^{k-1}Q\arrow[dl,"\tau_{T^{k-1}Q}"]\arrow[dr,"T\tau^{k-1}_Q"]  &&S^{\gamma}\arrow[ll, hook', "i"']\\
& & & &T^{k-1}Q\arrow[d,"\tau_Q^{k-1}"]\arrow[ul,"\gamma",bend left=20]
 & &TQ\arrow[dll,"\tau_Q"'] &\\
 & & & &Q & &\mathbb{R}\arrow[ll,"\phi"']\arrow[u,"j^1\phi"']  &\\
 \end{tikzcd}
\end{center}
where $\tau_Q^{k-1}:T^{k-1}Q\rightarrow Q$.

{\it Remark:} Recall that $E$ is implicit, so there are several vectors in $E$ projecting to the same point. The role of $\sigma$ is to reduce the unknown number to one. We require that the domain of the section is included in the intersection of Im$(\gamma)$ and $C$. Since for implicit systems $C$ may not be the whole $T^{*}T^{k-1}Q$, as a result we arrive at a vector field $X_{\sigma}$ that will satisfy a Hamilton equation of type
\begin{equation}
\iota_{X_{\sigma}}\Omega_{T^{k-1}Q}=\Theta(\gamma(q))
\end{equation}
for an arbitrary covector $\Theta$ defined at a point $\gamma(q)$.

 The construction of these local vector field using $\sigma$ would imply the following diagram

\begin{equation}\label{Xg}
  \xymatrix{ T^{*}T^{k-1}Q
\ar[dd]^{\pi_{T^{k-1}Q}} \ar[rrr]^{X_{\sigma}}&   & &TT^{*}T^{k-1}Q\ar[dd]^{T_{\pi_{T^{k-1}Q}}} &S\ar[l]^i\\
  &  & & &\\
T^{k-1}Q\ar@/^2pc/[uu]^{\gamma}\ar[rrr]^{X_{\sigma}^{\gamma}} &  & & TT^{k-1}Q &S^{\gamma}\ar[l]^i}
\end{equation}

\noindent
Explicitly, the locally constructed vector fields $X_{\sigma}\in \mathfrak{X}(T^{*}T^{k-1})Q$ and $X_{\sigma}^{\gamma}\in \mathfrak{X}(T^{k-1}Q)$ in coordinates would read:

\begin{equation}\label{coordfields}
X_{\sigma}=\sigma_{(\kappa)}^A(q_{(\kappa)}, \gamma^{(\kappa)})\frac{\partial }{\partial q_{(\kappa)}^A}+ \sigma^{(\kappa)}_A(q_{(\kappa)}, \gamma^{(\kappa)})\frac{\partial }{\partial p^{(\kappa)}_A}, \quad X_{\sigma}^{\gamma}=\sigma_{(\kappa)}^A(q_{(\kappa)}, \gamma^{(\kappa)})\frac{\partial }{\partial q_{(\kappa)}^A}
\end{equation}

If we use the one-form $\gamma:T^{k-1}Q\rightarrow T^{*}T^{k-1}Q$ and define the projected vector field
\begin{equation}
X_{\sigma}^{\gamma}=T_{\pi_{T^{k-1}Q}}\circ X_{\sigma}\circ \gamma,
\end{equation}
we have the following theorem.
\begin{theorem}[Implicit HJ theorem with an auxiliary section]
The one-form $\gamma$ will be a solution of an implicit higher order Hamilton--Jacobi problem if it satisfies the following relation
\begin{equation}
\sigma_{(\kappa)}^A(q_{(\kappa)}^A, \gamma^{(\kappa)}_A(q_{(\kappa)}))\frac{\partial \gamma^{(\kappa)}_A }{\partial q_{(\kappa)}^A}=\sigma^{(\kappa)}_A(q_{(\kappa)}^A, \gamma^{(\kappa)}_A(q_{(\kappa)})),
\end{equation}
when $\sigma$ is an auxiliary section $\sigma:T^{*}T^{k-1}Q\rightarrow TT^{*}T^{k-1}Q$. It is fulfilled that $\sigma^{-1}(S)=C$. Recall that since $S$ is an implicit submanifold, it does not necessarily project on the whole $T^{*}T^{k-1}Q$, but in a submanifold $C$ of it.
\end{theorem}

\begin{proof}
It is straightforward using that
\begin{equation}
T_{\gamma}(X_{\sigma}^{\gamma})=X_{\sigma}\circ \gamma
\end{equation}
and the expressions of $X_{\sigma}^{\gamma}$ and $X_{\sigma}$ in coordinates as in \eqref{coordfields}.
\end{proof}

\section{Applications}

\subsection{A (homogeneous) deformed elastic cylindrical beam with fixed ends}

Let $Q$ be a one-dimensional manifold with coordinate $q_{(0)}$, and introduce the second order Lagrangian
\begin{equation} \label{Lag-ex-1}
L(q_{(0)},q_{(1)},q_{(2)})=\frac{1}{2}\mu q_{(2)}^2+\rho q_{(0)}
\end{equation}
in terms of a local coordinate system $(q_{(0)},q_{(1)},q_{(2)})$ on $T^2Q$. 

\bigskip

\noindent \textbf{The Morse family method - Ostrogradsky Momenta}. 
We will first apply the Ostrogradsky method. In this method, the corresponding energy function is computed to be
\begin{equation}\label{morsefamily1}
E(q_{(0)},q_{(1)},q_{(2)},p^{(0)},p^{(1)})=p^{(0)}q_{(1)}+p^{(1)}q_{(2)}-\frac{1}{2}\mu q_{(2)}^2-\rho q_{(0)},
\end{equation}
where $q_{(2)}\in \mathbb{R}$ is the fiber component and $(q_{(0)},q_{(1)},p^{(0)},p^{(1)})$ are the canonical coordinates on $T^*TQ$. Here, $q_{(2)}$ is a Lagrange multiplier. The Morse family $E$ generates a Lagrangian submanifold $S$ of $TT^*TQ$, that corresponds with 
$$
S=\{(q_{(0)},q_{(1)},p^{(0)},\mu q_{(2)}; q_{(1)},q_{(2)},\rho,-p^{(0)})\in TT^*TQ:q_{(2)}\in \mathbb{R}\}.
$$
This Lagrangian submanifold defines the following differential equation
\begin{equation}\label{fourthorder}
 \ddddot{q_{(0)}}=-\frac{\rho}{\mu},
 \end{equation}
 which is exactly the second order Euler-Lagrange equation generated by the Lagrangian function $L$.  
The projection of $S$ onto the cotangent bundle $T^*TQ$ results in the submanifold
$$
C=\{(q_{(0)},q_{(1)};p^{(0)},p^{(1)})\in T^*TQ: p^{(1)}=\mu q_{(2)}\in \mathbb{R}\}.
$$
Let us now consider a closed one-form $\gamma=\gamma^{(0)} dq_{(0)}+\gamma^{(1)} dq_{(1)}$ and write the Hamilton-Jacobi equations (\ref{HJfor2ndOr}).
\begin{eqnarray}\label{HJfor2ndOr}
\begin{cases}
q_{(1)}\frac{\partial \gamma^{(0)}}{\partial q_{(0)}}+q_{(2)} \frac{\partial \gamma^{(1)}}{\partial q_{(0)}}-\rho=0 \\
\gamma^{(0)}+q_{(1)}\frac{\partial \gamma^{(0)}}{\partial q_{(1)}}+q_{(2)} \frac{\partial \gamma^{(1)}}{\partial q_{(1)}}=0 \\
\gamma^{(1)}-\mu q_{(2)}=0.
\end{cases}
\end{eqnarray}
If we substitute the last equation into the Morse family \eqref{morsefamily1} equal to constant, and we assume that $\gamma=dW$ for some real valued function $W$ on $TQ$, we arrive at that
\begin{eqnarray}\label{HJfor2ndOr-}
q_{(1)}\frac{\partial W}{\partial q_{(0)}}+\frac{1}{2\mu} \left(\frac{\partial W}{\partial q_{(1)}}\right)^2-\rho q_{(0)}=0.
\end{eqnarray}
Note that, in this case, solving the Hamilton-Jacobi equation is much more difficult than solving \eqref{fourthorder}.

\bigskip

\noindent \textbf{The Morse family method - The Schmidt's method}. Let us now propose the Schmidt method (\ref{S-2-HJ}). In this case, we have a two-dimensional acceleration bundle $AQ$ with coordinates $(q_{(0)},a_{(0)})$. Its tangent bundle $TAQ$ is four-dimensional with coordinates $(q_{(0)},a_{(0)};q_{(1)},a_{(1)})$. We pull back the Lagrangian in (\ref{Lag-ex-1}) to $TAQ$ by means of the isomorphism (\ref{iso}) which reads that 
$$ L=\frac{1}{2}\mu a_{(0)}^2+\rho q_{(0)}.$$ 
The compatibility condition (\ref{chi}) and the non-degeneracy of the Lagrangian suggests the auxiliary function $F=-\mu a_{(0)}q_{(1)}$. So, the extended Lagrangian (\ref{L2}) turns out to be
$$ L_2=\rho q_{(0)}-\frac{1}{2}\mu a_{(0)}^2-\mu a_{(1)}q_{(1)}.$$ 
The dual coordinates on the cotangent bundle $T^*AQ$ is given by $(q_{(0)},a_{(0)};p^{(0)},\pi^{(0)})$. The conjugate momenta is computed to be $\pi^{(0)}=-\mu q_{(1)}$. According to (\ref{canH}), this results with the following Hamiltonian function
$$ H=\frac{1}{2}\mu a_{(0)}^2-\frac{1}{\mu}\pi^{(0)} p^{(0)}-\rho q_{(0)}.$$ 
To arrive at the Hamilton-Jacobi equation, assume a closed one-form $\gamma=\gamma d q_{(0)} + \alpha d a_{(0)}$ defined on the acceleration bundle $AQ$, that is $\frac{\partial \gamma}{\partial a_{(0)}}=\frac{\partial \alpha}{\partial q_{(0)}}$. Recalling the Hamilton-Jacobi theorem asserts that the restriction of $H$ on $\gamma$ is constant \eqref{HJfor2ndOr}. Taking exterior derivative of this, we have the following set of equations
\begin{eqnarray}
\frac{\partial \alpha}{\partial a_{(0)}}\alpha +\frac{\partial \alpha}{\partial a_{(0)}}\gamma = \mu^2 a_{(0)}, \\
 \frac{\partial \gamma}{\partial q_{(0)}} \alpha +\frac{\partial \alpha}{\partial q_{(0)}} \gamma=- 2\rho.
\end{eqnarray}
Note that, this Hamilton-Jacobi problem reduces to the one studied in  (\ref{example}) if $\rho=0$. In this case, we solve the system as $\gamma=c_2$ and $\alpha=\frac{1}{c_2}\mu^2  a_{(0)} + c$, where $c$ and $c_2$ are constants.

\subsection{One dimensional version of the end of a javelin}
Let us consider the following Lagrangian on $T^2Q$ of the one-dimensional manifold $Q$ equipped with $(q_{(0)},q_{(1)},q_{(2)})$ given by
\begin{equation} \label{Lag-ex-2}
L(q_{(0)},q_{(1)},q_{(2)})=\frac{1}{2} q_{(1)}^2-\frac{1}{2} q_{(2)}^2.
\end{equation}

\bigskip

\noindent \textbf{The Morse family method - Ostrogradsky Momenta}. 
The associated energy function is given by
$$E(q_{(0)},q_{(1)},q_{(2)},p^{(0)},p^{(1)})=p^{(0)}q_{(1)}+p^{(1)}q_{(2)}+\frac{1}{2} q_{(2)}^2-\frac{1}{2} q_{(1)}^2.$$
Here, $q_{(2)}\in \mathbb{R}$ is the fiber component and $(q_{(0)},q_{(1)},p^{(0)},p^{(1)})$ are the canonical coordinates on $T^*TQ$. The Morse function $E$ generates the Lagrangian submanifold of $TT^*TQ$ given by
$$
S=\{(q_{(0)},q_{(1)},p^{(0)},-q_{(2)}; q_{(1)},q_{(2)},0,p^{(0)}-q_{(1)})\in TT^*TQ:q_{(2)}\in \mathbb{R}\}.
$$
This Lagrangian submanifold defines the equations
$$ \dddot{q_{(0)}}+\ddot{q_{(0)}}=c,$$
where $c$ is a constant. The projection of $S$ onto the cotangent bundle $T^*TQ$ is a three dimensional manifold
$$
C=\{(q_{(0)},q_{(1)};p^{(0)},p^{(1)})\in T^*TQ: p^{(1)}=- q_{(2)}\in \mathbb{R}\}.
$$
for a fixed $q_{(2)}$. For a closed one-form $\gamma^{(0)} dq_{(0)}+\gamma^{(1)} dq_{(1)}$, the Hamilton-Jacobi equation according to Theorem \eqref{HJfor2ndOr} turns out to be
\begin{eqnarray} \label{gamma-ex}
\begin{cases}
q_{(1)} \frac{\partial \gamma^{(0)}}{\partial q_{(0)}}+q_{(2)} \frac{\partial \gamma^{(1)}}{\partial q_{(0)}}=0
\\
\gamma^{(0)}+q_{(1)} \frac{\partial \gamma^{(0)}}{\partial q_{(1)}}+q_{(2)} \frac{\partial \gamma^{(1)}}{\partial q_{(1)}}-q_{(1)}=0
\\
\gamma^{(1)}+ q_{(2)}=0.
\end{cases}
\end{eqnarray}
We can solve $q_{(2)}$ from the last equation and if we substitute it in the equation $E=cst$, under the image of $\gamma=dW$ for some real valued function $W$ on $TQ$, we arrive at
\begin{eqnarray}\label{HJfor2ndOr-}
\frac{\partial W}{\partial q_{(0)}} q_{(1)} -\frac{1}{2} \left(\frac{\partial W}{\partial q_{(1)} }\right)^2 -\frac{1}{2} q_{(1)}^2=0.
\end{eqnarray}
There is a solution \cite{Co84}
$$ W(q_{(0)},q_{(1)})=Aq_{(0)}+\sqrt{2}\int\sqrt{Aq_{(1)}-\frac{1}{2}q_{(1)}^2-B}dq_{(1)}$$
which results with a one-form $\gamma$ solving the system (\ref{gamma-ex}) in form 
$$ \gamma=A dq_{(0)}+   \sqrt{2} \sqrt{\left(Aq_{(1)}-\frac{1}{2}q_{(1)}^2-B\right)}d q_{(1)}.$$

\bigskip

\noindent \textbf{The Morse family method - The Schmidt's method}. As an alternative realization of the Hamilton-Jacobi problem, we can use the Schmidt method in (\ref{S-2-HJ}). As in the previous subsection, we assume that acceleration bundle $AQ$ is two-dimensional with local coordinates $(q_{(0)},a_{(0)})$, and $TAQ$ is a four-dimensional manifold with $(q_{(0)},a_{(0)};q_{(1)},a_{(1)})$. We pull back the Lagrangian $L$ in (\ref{Lag-ex-2}) by means of the isomorphism \eqref{iso} and arrive at that
\begin{equation}
L=\frac{1}{2} q_{(1)}^2-\frac{1}{2} a_{(0)}^2.
\end{equation}
In this case, the auxiliary function is taken to be $F= a_{(0)}q_{(1)}$. Note that, $F$ satisfies the compatibility condition in (\ref{chi}). So, the first order Lagrangian function (\ref{L2}) is computed to be
$$ L_2= \frac{1}{2} q_{(1)}^2 + \frac{1}{2} a_{(0)}^2 + q_{(1)} a_{(1)}.$$ 
The coordinates on the cotangent bundle $T^*AQ$ are $(q_{(0)},a_{(0)};p^{(0)},\pi^{(0)})$ and the conjugate momenta is computed to be $\pi^{(0)}= q_{(1)}$. This results in the following Hamiltonian function
$$ H=\pi^{(0)} p^{(0)}-\frac{1}{2} (\pi^{(0)})^2 - \frac{1}{2} a_{(0)}^2$$ 
The Hamilton-Jacobi theorem in the acceleration bundle \eqref{HJacc} asserts that the restriction of $H$ on a closed one-form $dW$ is constant, say $c$. See that this can be written as 
$$
\frac{ \partial W}{\partial a_{(0)}}\frac{ \partial W}{\partial q_{(0)}}
- \left( \frac{ \partial W}{\partial a_{(0)}}\right)^2 - \frac{1}{2} a_{(0)}^2=c.
$$ 
A solution of this equation can easily be computed to be 
$$W=\frac{1}{\sqrt{2}}ln\left(a_{(0)}+\sqrt{a_{(0)}^2+2c} \right)
+\frac{1}{2\sqrt{2}}a_{(0)}\sqrt{a_{(0)}^2+2c}.$$ 

\subsection{A simple degenerate model}

Now we consider $Q$ as a three dimensional manifold with coordinates $(x,y,z)$ and consider the following degenerate second order Lagrangian
\begin{equation}
L=\frac{1}{2}(\ddot{x}+\ddot{y})^2.
\end{equation}  

\bigskip

\noindent \textbf{The Morse family method - Ostrogradsky Momenta}. On the cotangent bundle $T^*TQ$, we introduce the momenta $(p_x,p_y,p_z;p_{\dot{x}},p_{\dot{y}},p_{\dot{z}})$ and the energy function
\begin{equation}
E=p_x\dot{x}+p_y\dot{y}+p_z\dot{z}+p_{\dot{x}}\ddot{x}+p_{\dot{y}}\ddot{y}+p_{\dot{z}}\ddot{z}-\frac{1}{2}(\ddot{x}+\ddot{y})^2.
\end{equation}
Assume a function $W$ depending on $(x,y,z;\dot{x},\dot{y},\dot{z})$, then the Hamilton-Jacobi problem (\ref{HJfor2ndOr2--}) read 
\begin{equation}
\begin{split}
\frac{\partial ^2 W}{\partial x \partial x} \dot{x}
+
\frac{\partial ^2 W}{\partial x \partial y} \dot{y}
+
\frac{\partial ^2 W}{\partial x \partial z} \dot{z}
+
\frac{\partial ^2 W}{\partial x \partial \dot{x}} \ddot{x}
+
\frac{\partial ^2 W}{\partial x \partial \dot{y}} \ddot{y}
+
\frac{\partial ^2 W}{\partial x \partial \dot{z}} \ddot{z}
=0
\\
\frac{\partial W}{\partial x }
+
\frac{\partial ^2 W}{\partial \dot{x} \partial x} \dot{x}
+
\frac{\partial ^2 W}{\partial \dot{x} \partial y} \dot{y}
+
\frac{\partial ^2 W}{\partial \dot{x} \partial z} \dot{z}
+
\frac{\partial ^2 W}{\partial \dot{x} \partial \dot{x}} \ddot{x}
+
\frac{\partial ^2 W}{\partial \dot{x} \partial \dot{y}} \ddot{y}
+
\frac{\partial ^2 W}{\partial \dot{x} \partial \dot{z}} \ddot{z}
=0
\\
\frac{\partial ^2 W}{\partial y \partial x} \dot{x}
+
\frac{\partial ^2 W}{\partial y \partial y} \dot{y}
+
\frac{\partial ^2 W}{\partial y \partial z} \dot{z}
+
\frac{\partial ^2 W}{\partial y \partial \dot{x}} \ddot{x}
+
\frac{\partial ^2 W}{\partial y \partial \dot{y}} \ddot{y}
+
\frac{\partial ^2 W}{\partial y \partial \dot{z}} \ddot{z}
=0
\\
\frac{\partial W}{\partial y }
+
\frac{\partial ^2 W}{\partial \dot{y} \partial x} \dot{x}
+
\frac{\partial ^2 W}{\partial \dot{y} \partial y} \dot{y}
+
\frac{\partial ^2 W}{\partial \dot{y} \partial z} \dot{z}
+
\frac{\partial ^2 W}{\partial \dot{y} \partial \dot{x}} \ddot{x}
+
\frac{\partial ^2 W}{\partial \dot{y} \partial \dot{y}} \ddot{y}
+
\frac{\partial ^2 W}{\partial \dot{y} \partial \dot{z}} \ddot{z}
=0
\\
\frac{\partial ^2 W}{\partial z \partial x} \dot{x}
+
\frac{\partial ^2 W}{\partial z \partial y} \dot{y}
+
\frac{\partial ^2 W}{\partial z \partial z} \dot{z}
+
\frac{\partial ^2 W}{\partial z \partial \dot{x}} \ddot{x}
+
\frac{\partial ^2 W}{\partial z \partial \dot{y}} \ddot{y}
+
\frac{\partial ^2 W}{\partial z \partial \dot{z}} \ddot{z}
=0
\\
\frac{\partial W}{\partial z }
+
\frac{\partial ^2 W}{\partial \dot{z} \partial x} \dot{x}
+
\frac{\partial ^2 W}{\partial \dot{z} \partial y} \dot{y}
+
\frac{\partial ^2 W}{\partial \dot{z} \partial z} \dot{z}
+
\frac{\partial ^2 W}{\partial \dot{z} \partial \dot{x}} \ddot{x}
+
\frac{\partial ^2 W}{\partial \dot{z} \partial \dot{y}} \ddot{y}
+
\frac{\partial ^2 W}{\partial \dot{z} \partial \dot{z}} \ddot{z}
=0.
\end{split}
\end{equation} 
Although this system looks cumbersome, the set of constraints in (\ref{HJfor2ndOr2--}) is simply computed as  
\begin{equation}
\frac{\partial W}{\partial \dot{x}}=\frac{\partial W}{\partial \dot{y}}=\ddot{x} + \ddot{y}, \qquad \frac{\partial W}{\partial \dot{z}}=0,
\end{equation}
what reduces this huge system to a more reasonable one. For example, the independece of $W$ to $\dot{z}$ from the last line of the system gives the independence of $W$ to $z$. So that we have actually 4 number of equations. Then the first two constraints lead to the following reduced system of equations
\begin{equation}
\begin{split}
\frac{\partial ^2 W}{\partial x \partial x} \dot{x}
+
\frac{\partial ^2 W}{\partial x \partial y} \dot{y}
+
\frac{\partial ^2 W}{\partial x \partial \dot{x}} \frac{\partial W}{\partial \dot{x}}
=0
\\
\frac{\partial W}{\partial x }
+
\frac{\partial ^2 W}{\partial \dot{x} \partial x} \dot{x}
+
\frac{\partial ^2 W}{\partial \dot{x} \partial y} \dot{y}
+
\frac{\partial ^2 W}{\partial \dot{x} \partial \dot{x}} \frac{\partial W}{\partial \dot{x}}
=0
\\
\frac{\partial ^2 W}{\partial y \partial x} \dot{x}
+
\frac{\partial ^2 W}{\partial y \partial y} \dot{y}
+
\frac{\partial ^2 W}{\partial y \partial \dot{x}} \frac{\partial W}{\partial \dot{x}}
=0
\\
\frac{\partial W}{\partial y }
+
\frac{\partial ^2 W}{\partial \dot{y} \partial x} \dot{x}
+
\frac{\partial ^2 W}{\partial \dot{y} \partial y} \dot{y}
+
\frac{\partial ^2 W}{\partial \dot{y} \partial \dot{x}} \frac{\partial W}{\partial \dot{x}}
=0.
\end{split}
\end{equation} 
Notice that a solution of this can easily be noticed as 
\begin{equation}
W=a\dot{x}+b\dot{y}.
\end{equation}
\subsection{Second order Lagrangian systems with affine dependence on the acceleration}

In this subsection we are employing the theoretical parts presented in the previous section to the particular case of second order Lagrangian theories with affine dependence on the acceleration. To this end, we first define the following generic Lagrangian function 
\begin{equation} \label{Lag-Affine}
{L}(  q_{\left( 0\right) }, q_{\left( 1\right)
}, q_{\left( 2\right)})=f_A(q_{\left( 0\right) }, q_{\left( 1\right)
})q_{\left( 2\right)}^A+g(q_{\left( 0\right) }, q_{\left( 1\right)
})
\end{equation} 
on the second order tangent bundle $T^2Q$ where $f_A$ and $g$ are functions depending only on the position and the velocity. 

\bigskip

\noindent \textbf{The Morse family method - Ostrogradsky Momenta}.
Let us start with the first approach by introducing the Ostrogradsky momenta $(p^{(0)}_A,p^{(1)}_A)$ as the fiber coordinates of $T^*TQ$. Then the energy function take the form
\begin{equation}\label{}
E=p^{(0)}_Aq_{(1)}^A+p^{(1)}_Aq_{(2)}^A-L=p^{(0)}_Aq_{(1)}^A+p^{(1)}_Aq_{(2)}^A-f_A(q_{\left( 0\right) }, q_{\left( 1\right)
})q_{\left( 2\right)}^A-g(q_{\left( 0\right) }, q_{\left( 1\right)
}).
\end{equation} 
Now, let us introduce an exact one-form
\begin{equation}
\gamma=dW(q_{\left( 0\right) },q_{\left( 1\right) })=\frac{\partial W}{\partial q_{\left( 0\right) }^A}dq_{\left( 0\right) }^A+\frac{\partial W}{\partial q_{\left( 1\right) }^A}dq_{\left( 1\right) }^A
\end{equation}
 as given \eqref{gamma-k-1} and study the system of Hamilton-Jacobi equations (\ref{HJfor2ndOr2--}). In this case we have that 
\begin{align} \label{HJ-2-Affine}
\begin{aligned}
\frac{\partial^2 W} {\partial q^B_{(0)} \partial q_{\left( 0\right) }^A} {q_{(1)}^A}
+
\frac{\partial^2 W} {\partial q^B_{(0)} \partial q_{\left( 1\right) }^A} {q_{(2)}^A}
&=\frac{\partial f_A}{\partial q^B_{(0)}}q^A_{(2)}+\frac{\partial g}{\partial q^B_{(0)}}
\\
\frac{\partial W} {\partial q_{\left( 0\right) }^B} 
+
\frac{\partial^2 W} {\partial q^B_{(1)} \partial q_{\left( 0\right) }^A} {q_{(1)}^A}
+
\frac{\partial^2 W} {\partial q^B_{(1)} \partial q_{\left( 1\right) }^A} {q_{(2)}^A}
&=\frac{\partial f_A}{\partial q^B_{(1)}}q^A_{(2)}+\frac{\partial g}{\partial q^B_{(1)}}
\\
\frac{\partial W}{\partial q^A_{(1)}}&= f_A(q_{(0)},q_{(1)}).
\end{aligned}
\end{align}
Consider now the third equation in the system (\ref{HJ-2-Affine}). Taking the partial derivative of this with respect to $q^B_{(1)}$ result with the following equality
\begin{equation}
\frac{\partial^2 W}{\partial q^A_{(1)}\partial q^B_{(1)}}= \frac{\partial f_A}{\partial q^B_{(1)}}(q_{(0)},q_{(1)}).
\end{equation}
Notice that the left hand side is symmetric with respect to the indices $A$ and $B$ whereas this is not generally true for an arbitrary functions $f_A$. This a first restriction to the application of the former theory. Even though there are numerous physical systems satisfying this symmetry criteria in the literature. There are also interesting physical models involving affine terms violating this symmetry. We provide two example important examples for such kind of systems in the conclusions section by pointing out some possible future works.

We can further investigate more on the integrability of the Hamilton-Jacobi equations.
To this end, we substitute the last line of the system (\ref{HJ-2-Affine}) into the first two equations. This reads 
\begin{align} \label{HJ-2-Affine-2}
\begin{aligned}
\frac{\partial^2 W} {\partial q^B_{(0)} \partial q_{\left( 0\right) }^A} {q_{(1)}^A}
&= \frac{\partial g}{\partial q^B_{(0)}}
\\
\frac{\partial W} {\partial q_{\left( 0\right) }^B} 
&=\frac{\partial g}{\partial q^B_{(1)}}-\frac{\partial f_A}{\partial q^B_{(0)}}q^A_{(1)}
\end{aligned}
\end{align}
Taking the partial derivative of the second line with respect to $\partial q_{\left( 0\right) }^A$ , multiplying by ${q_{(1)}^A}$, we arrive at the following differential equation
\begin{equation}
\frac{\partial g}{\partial q^B_{(0)}}-\frac{\partial^2 g}{\partial q^B_{(1)}\partial q^A_{(0)}}q_{\left( 1\right) }^A+ \frac{\partial^2 f_C}{\partial q^B_{(0)}\partial q_{\left( 0\right)}^A}q^C_{(1)}q_{\left( 1\right)}^A=0.
\end{equation}  
This is an integrability criterion for the HJ problem for second order Lagrangian fomalisms that are affine in acceleration. Assuming that this holds, the Hamilton-Jacobi problem can be written in a relatively easy form
\begin{align} \label{HJ-Affine-3}
\begin{aligned}
\frac{\partial W} {\partial q_{\left( 0\right) }^B} 
&=\frac{\partial g}{\partial q^B_{(1)}}-\frac{\partial f_A}{\partial q^B_{(0)}}q^A_{(1)}, \\
\frac{\partial W}{\partial q^A_{(1)}}&= f_A(q_{(0)},q_{(1)}).
\end{aligned}
\end{align} 

\noindent \textbf{The Morse family method - The Schmidt's method}. In this case, we shall start with the Lagrangian function (\ref{Lag-Affine}) once more, but in this case we investigate the associated Hamilton-Jacobi problem by means of the Schmidt Legendre transformation in the framework of the acceleration bundle. By choosing the auxiliary function $F=\delta_{AB} q^A_{(1)}m^B_{(0)}$ we write the equivalent Lagrangian function exhibited in (\ref{L-2-deg}) as follows
\begin{equation}
L_{2-deg}=f_A(q_{\left( 0\right) }, q_{\left( 1\right)
})a_{\left(0\right)}^A+g(q_{\left( 0\right) }, q_{\left( 1\right)
})+\delta_{AB} q^A_{(1)}m^B_{(1)}+\delta_{AB} a^A_{(0)}m^B_{(0)}
\end{equation}
which depends on the base components $(q_{\left( 0\right) },a_{\left( 0\right) },m_{\left( 0\right) })$ along with the velocities $(q_{\left( 1\right) },a_{\left( 1\right) },m_{\left( 1\right) })$. In this case, the energy function (\ref{totHam----}) turns out to be 
$$
E= p^{(0)}_A {q}_{(1)}^A+\pi^{(0)}_A{a}_{(1)}^A+\mu^{(0)}_A {m}_{(1)}^A-f_A(q_{\left( 0\right) }, q_{\left( 1\right)
})a_{\left(0\right)}^A-g(q_{\left( 0\right) }, q_{\left( 1\right)
})-\delta_{AB} q^A_{(1)}m^B_{(1)}-\delta_{AB} a^A_{(0)}m^B_{(0)}
$$
Assuming an exact one-form
\begin{equation}
\gamma=dW(q_{\left( 0\right) },a_{\left( 0\right) },m_{\left( 0\right) })=\frac{\partial W}{\partial q_{\left( 0\right) }^A}dq_{\left( 0\right) }^A+\frac{\partial W}{\partial a_{\left( 0\right) }^A}da_{\left( 0\right) }^A+\frac{\partial W}{\partial m_{\left( 0\right) }^A}dm_{\left( 0\right) }^A
\end{equation}
the Hamilton-Jacobi equation (\ref{HJ-L-2-deg}) turns out to be
\begin{equation}
\begin{split}
\frac{\partial^2 W}{\partial q_{\left( 0\right) }^A \partial q_{\left( 0\right) }^B}  q_{\left( 1\right) }^A
+
\frac{\partial^2 W}{\partial a_{\left( 0\right) }^A \partial q_{\left( 0\right) }^B}  a_{\left( 1\right) }^A
+
\frac{\partial^2 W}{\partial m_{\left( 0\right) }^A \partial q_{\left( 0\right) }^B}  m_{\left( 1\right) }^A
&=
\frac{\partial f_A}{\partial q_{\left( 0\right) }^B}  a_{\left( 0\right) }^A 
+ \frac{\partial g}{\partial q_{\left( 0\right) }^B}
\\
\frac{\partial^2 W}{\partial q_{\left( 0\right) }^A \partial a_{\left( 0\right) }^B}  q_{\left( 1\right) }^A
+
\frac{\partial^2 W}{\partial a_{\left( 0\right) }^A \partial a_{\left( 0\right) }^B}  a_{\left( 1\right) }^A
+
\frac{\partial^2 W}{\partial m_{\left( 0\right) }^A \partial a_{\left( 0\right) }^B}  m_{\left( 1\right) }^A
&=
f_B+\delta_{AB}m^A_{\left( 0\right) }
\\
\frac{\partial^2 W}{\partial q_{\left( 0\right) }^A \partial m_{\left( 0\right) }^B}  q_{\left( 1\right) }^A
+
\frac{\partial^2 W}{\partial a_{\left( 0\right) }^A \partial m_{\left( 0\right) }^B}  a_{\left( 1\right) }^A
+
\frac{\partial^2 W}{\partial m_{\left( 0\right) }^A \partial m_{\left( 0\right) }^B}  m_{\left( 1\right) }^A &=\delta_{AB}a^A_{\left( 0\right) }
\\
\frac{\partial W}{\partial q_{\left( 0\right) }^A }
&=
\frac{\partial f_B}{\partial q_{\left( 1\right) }^A}
 a_{\left( 0\right) }^B
+
\frac{\partial g}{\partial q_{\left( 1\right) }^A}
+
 \delta_{AB}m_{\left( 1\right) }^B
 \\
 \frac{\partial W}{\partial a_{\left( 0\right) }^A }&=0
 \\
 \frac{\partial W}{\partial m_{\left( 0\right) }^A }&=\delta_{AB}  q_{\left( 1\right) }^B.
\end{split}
\end{equation}
From the fifth line we see that $W$ does not depend on $a_{(0)}$. So that, the second line determines the identity $f_B=-\delta_{AB}m^A_{(0)}$. From the fourth and sixth equations, we substitute the Lagrange multipliers $m_{\left( 1\right) }^B$ and $q_{\left( 1\right) }^B$, in to the rest of the equations and we arrive at the following reduced system
\begin{equation}
\begin{split}
\delta^{AC}\frac{\partial^2 W}{\partial q_{\left( 0\right) }^A \partial q_{\left( 0\right) }^B}   \frac{\partial W}{\partial m_{\left( 0\right) }^C }
+
\delta^{AC}\frac{\partial^2 W}{\partial m_{\left( 0\right) }^A \partial q_{\left( 0\right) }^B} 
\left(
\frac{\partial W}{\partial q_{\left( 0\right) }^C }
-\frac{\partial f_D}{\partial q_{\left( 1\right) }^C}
 a_{\left( 0\right) }^D
-
\frac{\partial g}{\partial q_{\left( 1\right) }^C}
\right)
&=
\frac{\partial f_A}{\partial q_{\left( 0\right) }^B}  a_{\left( 0\right) }^A 
+ \frac{\partial g}{\partial q_{\left( 0\right) }^B}
\\
\delta^{AC}
\frac{\partial^2 W}{\partial q_{\left( 0\right) }^A \partial m_{\left( 0\right) }^B}  
\frac{\partial W}{\partial m_{\left( 0\right) }^C }
+
\delta^{AC}\frac{\partial^2 W}{\partial m_{\left( 0\right) }^A \partial m_{\left( 0\right) }^B}  
\left(
\frac{\partial W}{\partial q_{\left( 0\right) }^C }
-\frac{\partial f_D}{\partial q_{\left( 1\right) }^C}
 a_{\left( 0\right) }^D
-
\frac{\partial g}{\partial q_{\left( 1\right) }^C}
\right) &=\delta_{AB}a^A_{\left( 0\right) }.
\end{split}
\end{equation}
In this case, we have arrived at a relatively complicated PDE system comparing with the Ostrogradsky method. Indeed, the choice of one of the two methods is important for resolving the equations.

\subsection{Third order Lagrangian systems with affine dependence on the acceleration}

In this subsection, in order to exhibit the application area of the theoretical framework we have proposed, we shall investigate possible Hamilton-Jacobi realization of some class of the third order singular Lagrangian systems involving affine dependence to the third order derivative terms in the following form 
 \begin{equation} \label{Lag-Affine-3}
{L}(  q_{\left( 0\right) }, q_{\left( 1\right)
}, q_{\left( 2\right)})=f_A(q_{\left( 1\right) }, q_{\left( 2\right)
})q_{\left( 3\right)}^A+g(q_{\left( 0\right) }, q_{\left( 1\right)
})
\end{equation} 
on the third order tangent bundle $T^3Q$.

\bigskip

\noindent \textbf{The Morse family method - Ostrogradsky Momenta}. The energy function generating the dynamics of the Lagrangian \eqref{Lag-Affine-3} is
\begin{equation}\label{E-Affine-3}
E=p^{(0)}_Aq_{(1)}^A+p^{(1)}_Aq_{(2)}^A+p^{(2)}_Aq_{(3)}^A-f_A(q_{\left( 1\right) }, q_{\left( 2\right)
})q_{\left( 3\right)}^A-g(q_{\left( 0\right) }, q_{\left( 1\right)
}),
\end{equation} 
where $(p^{(0)},p^{(1)},p^{(2)})$ are the conjugate momenta defining the fiber coordinates of the cotangent bundle $T^*T^2Q$. Consider an exact one-form on $T^2Q$ which is in coordinates given by
\begin{equation}
\gamma=dW(q_{\left( 0\right) },q_{\left( 1\right) },q_{\left( 2\right) })=\frac{\partial W}{\partial q_{\left( 0\right) }^A}dq_{\left( 0\right) }^A+\frac{\partial W}{\partial q_{\left( 1\right) }^A}dq_{\left( 1\right) }^A
+
\frac{\partial W}{\partial q_{\left( 2\right) }^A}dq_{\left( 2\right) }^A
\end{equation}
following \eqref{gamma-k-1}. We write the system of Hamilton-Jacobi equations (\ref{HJfor2ndOr2--}) as follows 
\begin{equation} \label{Affine-3-HJ}
\begin{split}
\frac{\partial ^2 W}{\partial q_{\left( 0\right) }^A \partial q_{\left( 0\right) }^B}
 q_{\left( 1\right) }^A
+
\frac{\partial ^2 W}{\partial q_{\left( 1\right) }^A \partial q_{\left( 0\right) }^B}
 q_{\left( 2\right) }^A
+
\frac{\partial ^2 W}{\partial q_{\left( 2\right) }^A \partial q_{\left( 0\right) }^B}
 q_{\left( 3\right) }^A
&=
\frac{\partial g}{\partial q_{\left( 0\right) }^B}
\\
\frac{\partial  W}{\partial q_{\left( 0\right) }^B}
+
\frac{\partial ^2 W}{\partial q_{\left( 0\right) }^A \partial q_{\left( 1\right) }^B}
 q_{\left( 1\right) }^A
+
\frac{\partial ^2 W}{\partial q_{\left( 1\right) }^A \partial q_{\left( 1\right) }^B}
 q_{\left( 2\right) }^A
+
\frac{\partial ^2 W}{\partial q_{\left( 2\right) }^A \partial q_{\left( 1\right) }^B}
 q_{\left( 3\right) }^A
 &=
 \frac{\partial f_A}{\partial q_{\left( 1\right) }^B}q_{\left( 3\right) }^A
 +
 \frac{\partial g}{\partial q_{\left( 1\right) }^B}
 \\
\frac{\partial  W}{\partial q_{\left( 1\right) }^B}
+
\frac{\partial ^2 W}{\partial q_{\left( 0\right) }^A \partial q_{\left( 2\right) }^B}
 q_{\left( 1\right) }^A
+
\frac{\partial ^2 W}{\partial q_{\left( 1\right) }^A \partial q_{\left( 2\right) }^B}
 q_{\left( 2\right) }^A
+
\frac{\partial ^2 W}{\partial q_{\left( 2\right) }^A \partial q_{\left( 2\right) }^B}
 q_{\left( 3\right) }^A
 &=
 \frac{\partial f_A}{\partial q_{\left( 2\right) }^B}q_{\left( 3\right) }^A
\\
\frac{\partial W}{\partial q_{\left( 2\right) }^B}&=f_B.
\end{split}
\end{equation}
Let us try to simplify this system. See that the last line reads that ${\partial W}/{\partial q_{\left( 2\right) }^B}$ is independent of $q_{(0)}$, and leads to the observation that $\partial f_B/\partial q_{\left( 2\right) }^A$ must be symmetric with respect to the indices $A$ and $B$. Substitution of the last identity in (\ref{Affine-3-HJ}) into the second and third lines we arrive at a fairly more simple system
\begin{equation}  \label{Affine-3-HJ-red}
\begin{split}
\frac{\partial W}{\partial q_{\left( 0\right) }^B}
&= 
 \frac{\partial g}{\partial q_{\left( 1\right) }^B}
+
\frac{\partial ^2 f_B}{\partial q_{\left( 1\right) }^A \partial q_{\left( 1\right) }^C}
 q_{\left( 2\right) }^A 
 q_{\left( 2\right) }^C, \\ \frac{\partial W}{\partial q_{\left( 1\right) }^B}&=-
\frac{\partial f_B}{\partial q_{\left( 1\right) }^C }
 q_{\left( 2\right) }^C, \\ \frac{\partial W}{\partial q_{\left( 2\right) }^B}&=f_B
\end{split}
\end{equation}
See that, this system is coupled with the first line of the system (\ref{Affine-3-HJ}). So that substitution of \eqref{Affine-3-HJ-red} into the first line of (\ref{Affine-3-HJ}) must be identically satisfied. This compatibility condition reads
\begin{equation}
 \frac{\partial ^2 g}{\partial q_{\left( 0\right) }^A q_{\left( 1\right) }^B}
  q_{\left( 1\right) }^A
+
\frac{\partial  ^3 f_B}{\partial q_{\left( 0\right) }^A \partial q_{\left( 1\right) }^D \partial q_{\left( 1\right) }^C}
 q_{\left( 2\right) }^D 
 q_{\left( 2\right) }^C
 q_{\left( 1\right) }^A
 -
\frac{\partial ^2 f_B}{\partial q_{\left( 1\right) }^C \partial q_{\left( 1\right) }^A}
 q_{\left( 2\right) }^C q_{\left( 2\right) }^A=
\frac{\partial g}{\partial q_{\left( 0\right) }^B}.
\end{equation}  
It is also important to know that one only needs to perform direct integration to find $W$ after the functions $f_A$ and $g$ are determined. But to do this, $f_A$ and $g$ can not be arbitrarily chosen since they have to satisfy the integrability conditions arising form the system. 

\section{Conclusions and comments}

In this paper, we are proposing a Hamilton--Jacobi theory for higher order Lagrangian formalisms. Our theory works well for all non-degenerate systems and a large class of degenerate theories. The implicit character of singular systems has been studied in two different forms: one is making use of Morse families that play the role of the Hamiltonian, and giving rise to Lagrangian submanifolds, equivalently to the image of $\gamma$, which denotes the solution of a Hamilton--Jacobi problem. The other method consists on the local construction of a vector field associated with the implicit equations and defined on a proper domain compatible with implicit character.
The higher order derivatives are studied through both the Ostrogradsky-Legendre and Schmidt-Legendre transformations. In the case of second order Lagrangians we have employed the acceleration bundle picture.

As a future work, we want to generalize this formalism in a proper way, which would enable us to work all degenerate higher order Lagrangian systems: singular higher order Lagrangians coming from the gravitational theory. We will mostly be interested in two examples. One is the {chiral oscillator} in two dimensions.
This oscillator accounts for mirror symmetry, and in the case of a non-relativistic oscillator with a Chern-Simons term (independent of the metric), we have the expression: 
\begin{equation} \label{Lag-CO}
L=-\lambda\epsilon_{AB} q_{(1)}^A q_{(2)}^A + \frac{m}{2} \delta_{AB}q_{(1)}^A q_{(1)}^B
\end{equation}
where $\lambda$ and $m$ are nonvanishing constants \cite{CrGoMoRo16}. Here, $\epsilon_{AB }$ is a skew-symmetric tensor with $\epsilon_{12}=1$. The Lagrangian (\ref{Lag-CO}) is quasi-invariant
under the Galilean transformations.
The second example is Cl\'{e}ment Lagrangian which is a second order degenerate Lagrangian function \cite{Cl94}. It is defined on the second order tangent bundle $T^2Q$ where $Q$ is a semi-Riemannian manifold equipped with the
Minkowskian metric $\theta=[\theta_{AB}]$ with $(+,-,-)$. The Cl\'{e}ment Lagrangian is given by
\begin{equation}
L=-\frac{m}{2}\zeta \theta_{AB} {q}_{(1)}^A {q}_{(1)}^B-\frac{2m\Lambda }{\zeta }+\frac{\zeta
^{2}}{2\mu m} \epsilon_{ABC}  {q}_{(0)}^A {q}_{(1)}^B {q}_{(2)}^C
\label{clemlag}
\end{equation}%
where $\zeta =\zeta (t)$ is a function that allows arbitrary
reparametrizations of the variable $t$, whereas $\Lambda $ and $1/2m$ are the
cosmological and Einstein gravitational constants, respectively. Here, $\epsilon_{ABC}$ is a skewsymmetric three tensor determining the triple product, so this Lagrangian falls into the category of Lagrangians depending on the acceleration linearly \cite{CrGoMoRo16}. For the Hamiltonian analysis of this singular theory, we cite \cite{CaEsGu18}.

\section*{Acknowledgments}

This work has been partially supported by MINECO Grants MTM2016-
76-072-P and the ICMAT Severo Ochoa projects SEV-2011-0087 and SEV2015-0554. 



\end{document}